%% file: ms.tex
\pdfoutput=1

\documentclass{article}
\usepackage[square,sort,comma,numbers]{natbib}
\usepackage{arxiv}

\usepackage[utf8]{inputenc} 
\usepackage[T1]{fontenc}    
\usepackage{hyperref}       
\usepackage{url}            
\usepackage{booktabs}       
\usepackage{amsfonts}       
\usepackage{nicefrac}       
\usepackage{microtype}      
\usepackage{lipsum}

\usepackage{float}
\usepackage{graphicx}
\usepackage{balance}  
\usepackage{bbm}

\usepackage[newfloat, frozencache]{minted}

\usepackage{caption}
\usepackage{comment}
\usepackage{xspace}
\usepackage{hyperref}

\usepackage{amsmath}
\usepackage{algorithm}
\usepackage[noend]{algpseudocode}
\usepackage{subcaption}
\usepackage{tabularx}
\usepackage{amsthm}

\usepackage{setspace}
\title{VizRec: a framework for secure data exploration via visual representation}
\newcommand{\email}[1]{{\texttt #1}}

\author{Lorenzo De Stefani \\
Brown University \\
Providence, RI \\
\email{lorenzo\_destefani@brown.edu}
\\[3ex]
\textbf{Tim Kraska} \\
Massachusetts Institute of Technology \\
Cambridge, MA \\
\email{kraska@mit.edu}
\And
Leonhard F. Spiegelberg \\
Brown University \\
Providence , RI \\ 
\email{leonhard\_spiegelberg@brown.edu} 
\\[3ex]
\textbf{Eli Upfal} \\
Brown University \\
Providence, RI \\
\email{eli@cs.brown.edu}
}

\SetupFloatingEnvironment{listing}{name=Source Code}

\newcommand{\system}{VizRec\xspace} 

\DeclareMathOperator*{\dom}{\mathrm{dom}}

\newcommand{\PrD}[2]{\textrm{Pr}_{#1}\left(#2\right)}

\newcommand{\chisqtest}{$\chi^2$-test~}

\newcommand{\mdataset}{\mathcal{D}}
\newcommand{\dataset}{$\mathcal{D}$}
\newcommand{\vis}[1]{\mathcal{V}_{#1}}

\newcommand{\Expe}[2]{E_{#1}\left[#2\right]}

\newtheorem{definition}{Definition}

\newtheorem{theorem}{Theorem}
\newtheorem{corollary}{Corollary}
\newtheorem{lemma}{Lemma}
\newtheorem{fact}{Fact}

\newmintinline{sql}{showspaces}

\begin{document}
\maketitle

\begin{abstract}
\onehalfspacing
Visual representations of data (visualizations) are tools of great importance and widespread use in data analytics as they provide users visual insight to patterns 
in the observed data in a simple and effective way.
However, since visualizations tools are applied to sample data, there is a a risk of visualizing random fluctuations in the sample rather than a true pattern in the data. This problem is even more significant when visualization is used to identify \emph{interesting} patterns among many possible possibilities, or to identify an interesting deviation in a pair of observations among many possible pairs,
as commonly done in \emph{visual recommendation systems}.

We present \emph{VizRec}, a framework for improving the performance of visual recommendation systems by quantifying the statistical significance of recommended visualizations. The proposed methodology allows to control the probability of misleading visual recommendations using both classical statistical testing procedures and a novel application of the Vapnik Chervonenkis (VC) dimension method which is a fundamental concept in statistical learning theory.
\end{abstract}

\keywords{Data Visualization, False Discoveries, Machine Learning, Multiple Comparison Problem}

\input{intro.tex}
\input{problem_statement.tex}
\input{new_sec3.tex}
\input{sec4.tex}
\input{vc.tex}
\input{tradeoff.tex}

\input{new_strategies.tex}

\input{experiments_new.tex}
\input{related.tex}
\input{conclusion.tex}

\input{acknowledgments.tex}
\bibliographystyle{abbrvnat}
\bibliography{refs} 
\newpage
\input{appendix.tex}

\end{document}

%% file: intro.tex
\section{Introduction}
\label{sec:intro}
Visual recommendation engines, such as SeeDB~\cite{vartak2014seedb}, Voyager2~\cite{wongsuphasawat2017voyager}, Rank-By-Feature~\cite{rankbyfeatue}, Show Me~\cite{showme}, MuVE~\cite{muve}, VizDeck~\cite{vizdeck}, DeepEye~\cite{deepeye}, or Draco-Learn \cite{draco}, aim to help users to more quickly explore a dataset and find interesting insights. 
To achieve that goal they use widely different approaches and techniques. 
For example, SeeDB~\cite{vartak2014seedb} makes recommendations based on a reference view; it tries to find a visualization which is very different from the one the user has currently on the screen. 
In contrast, DeepEye~\cite{deepeye} tries to generally recommend a good visualization for a given dataset based on previously generated visualizations. 

However, all these systems do have in common that they can significantly increase the risk of finding false insights. 
This happens in the moment a visualization is not just a pretty picture but a tool presenting facts about the data to the user. 
For example, consider a user exploring a dataset containing information about different wines.
After browsing the data for a bit, she creates a visualization of ranking by origin showing that wines from France are higher rated. 
If her only takeaway from the visual is, that in {\bf this particular dataset} wines from France have a higher rating, there is no risk of false insight. 
Essentially, in this case no inference happens as she is completely aware that the next dataset could look entirely different. 
However, it is neither in the nature of users to constraint themselves to such thinking~\cite{mcpvisual}, nor would in many cases such a visualization be very insightful. 
Rather, based on the visualization she most likely would infer that French wines are {\em generally} rated higher; {\em generalizing} her visualization insight to general datasets and thus creating an actually interesting insight.  
Statistically savvy users will now test this insight on whether this generalization is actually statistically valid using the appropriate test. 
Even more technically savvy users will also consider other hypothesis they tried and adjust the statistical testing procedure to account for the multiple comparisons problem.
This is important as every additional hypothesis, explicitly expressed as a test or implicitly observed through a visualization, increases the risk of finding insights which are just spurious effects. 

However, what happens when the visualization recommendations are generated by one of the above systems?
First and most importantly, the user does not know if the effect shown by the visualization is actually significant or not. 
Even worse, she can not use a standard statistical method and ``simply'' test  the effect shown in the visualization for significance. 
Visual recommendation engines are potentially checking thousands of visualizations for their interesting-factor (e.g., in case of SeeDB how different the visualization is to the current one) in just a few seconds.
As a result, by testing thousands of visualizations it is almost guaranteed that the system will find something ``\emph{interesting}'' regardless of whether the observed phenomenon is actually statistically valid or not.
A test for significance for the recommended visualization should therefore consider the whole history of tests done by the recommendation engine. 

Advocates of visual recommendation engines usually argue that visual recommendations systems are meant to be as hypothesis generation engines, which should always be validated on a separate hold-out dataset. 
While this is a valid method to control false discoveries, it is also important to understand its implications: 
(1) None, really none, of the found insights from the exploration dataset should be regarded as an actual insight before they are validated. 
This is clearly problematic if one observation may steer towards another during the exploration. 
(2) Splitting a dataset into an exploration and a hold-out set can significantly reduce the power (i.e., the chance to find actual true phenomena). 
(3) The hold-out needs to be controlled for the multi-hypothesis problem unless the user only wants to use it exactly once for a single test.

In this paper, we present an alternative approach, called \system, a framework to make visual recommendation engines ``safe''.
We focus on the visual recommendation technique proposed by SeeDB as it uses a clear semantic for what ``interesting'' means. 
However, our techniques can be adjusted to other visualization frameworks or even hypothesis generation tools, like Data Polygamy, as long as the ``interesting'' criterion can be expressed as a statistical test. 
The core idea of \system is that it not only evaluates the interesting-factor using a statistical test, but also that it automatically adjusts the significance value based on the search space of the recommendation engine to avoid the multiple-hypothesis pitfall. 
We make the following contributions:
\begin{itemize}
\item We formalize the process of making visualization recommendations as  statistical hypothesis testing.
\item We discuss how different possible approaches to make visualization recommendation engines safe based on classical statistical testing, and on the use of Chernoff-type large deviation bounds. We further discuss how the performance of both these approaches decreases due to the necessity of accounting for a high number of potentially adaptively chosen tests.
\item We propose a method based on the use of VC dimension, which allows controlling the probability of observing fake discoveries during the visualization recommendation process. \system{} allows control of the \emph{Family Wise Error Rate} (FWER) at a given level $\delta$.
\item We evaluate the performance of our system, in comparison with SeeDB via extensive experimental analysis. 
\end{itemize}
The remainder of this paper is organized as follows: In Secton~\ref{sec:prob-stat} we give a definition of the visualization recommendation problem in rigorous probabilistic terms. In Section~\ref{sec:statssafvis} we discuss possible approaches for the visualization recommendation problems, and why highlight how they both suffer due to the necessity of accounting for a high number of statistical tests. In Section~\ref{sec:VC} we introduce our \system{} approach based on VC dimension and we argue how it allows to overcome the problems discussed in Section~\ref{sec:statssafvis}. In Section~\ref{sec:discussion}, we discuss some guidelines of practical interest for implementation, while in Section~\ref{sec:experiment} we present an extensive experimental evaluation of the effectiveness of~\system{}.

%% file: problem_statement.tex
\section{Problem Statement}
\label{sec:prob-stat}
In this section we first describe informally how SeeDB~\cite{vartak2014seedb} makes visual recommendations, and then formalize the problem of providing statistically valid visualization recommendations. 

\subsection{SeeDB}
SeeDB makes recommendations based on the currently shown visualization, and the corresponding \emph{reference query}. 
SeeDB explores possible recommendations (\emph{target queries}) by most commonly adding/changing the composition of the reference query. 
To rank the recommendations, SeeDB recommends to the user the most interesting target queries based on the deviations.
Thus, SeeDB assumes a larger deviation indicates a more interesting target query.
SeeDB can use different types of measures to quantify the difference between reference and target visualization.
However, earth-mover distance or KL-divergence are prevalently used. 
Furthermore, SeeDB truncates uninteresting visualizations if the deviation value (e.g., KL-divergence) is below a certain threshold, which can be seen as a minimum visual distance. 

To showcase how SeeDB can recommend spurious correlations, we used a survey conducted on Amazon Mechanical Turk \cite{CIDRBrown} with  $2,644$ answers for $35$ (mostly unrelated) multiple-choice questions. 
Questions range from \emph{Would you rather drive an electric car or a gas-powered car?} to \emph{What is your highest achieved education level?} and \emph{Do you play Pokemon Go?}.
\begin{figure}[H]
\begin{center}
    \begin{subfigure}[b]{0.3\textwidth}
        \centering
        \includegraphics[width=\textwidth]{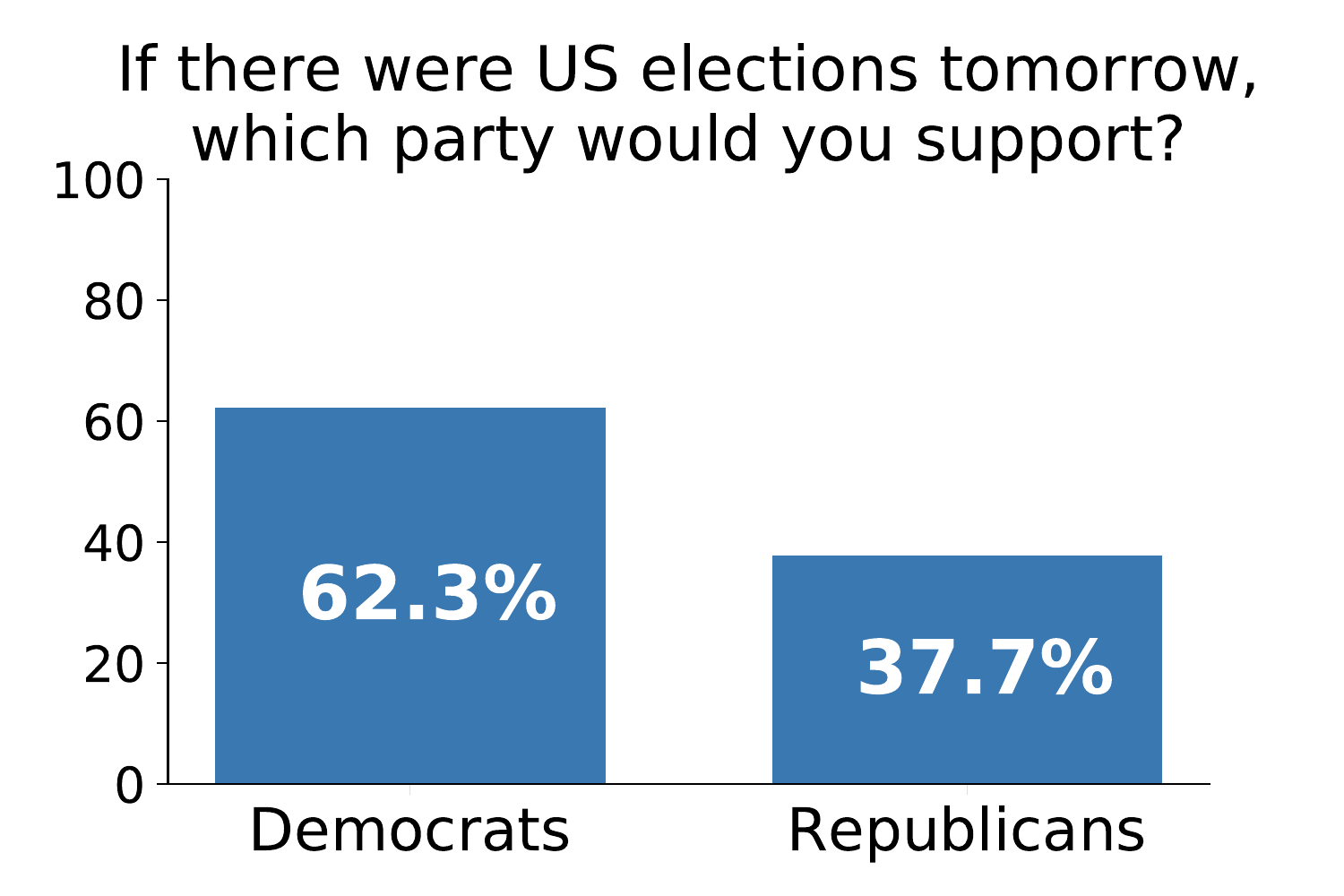}
        \caption{Reference View}
        \label{fig:seedb-mturk-example-refer}
    \end{subfigure}
    \begin{subfigure}[b]{0.3\textwidth}
        \centering
        \includegraphics[width=\textwidth]{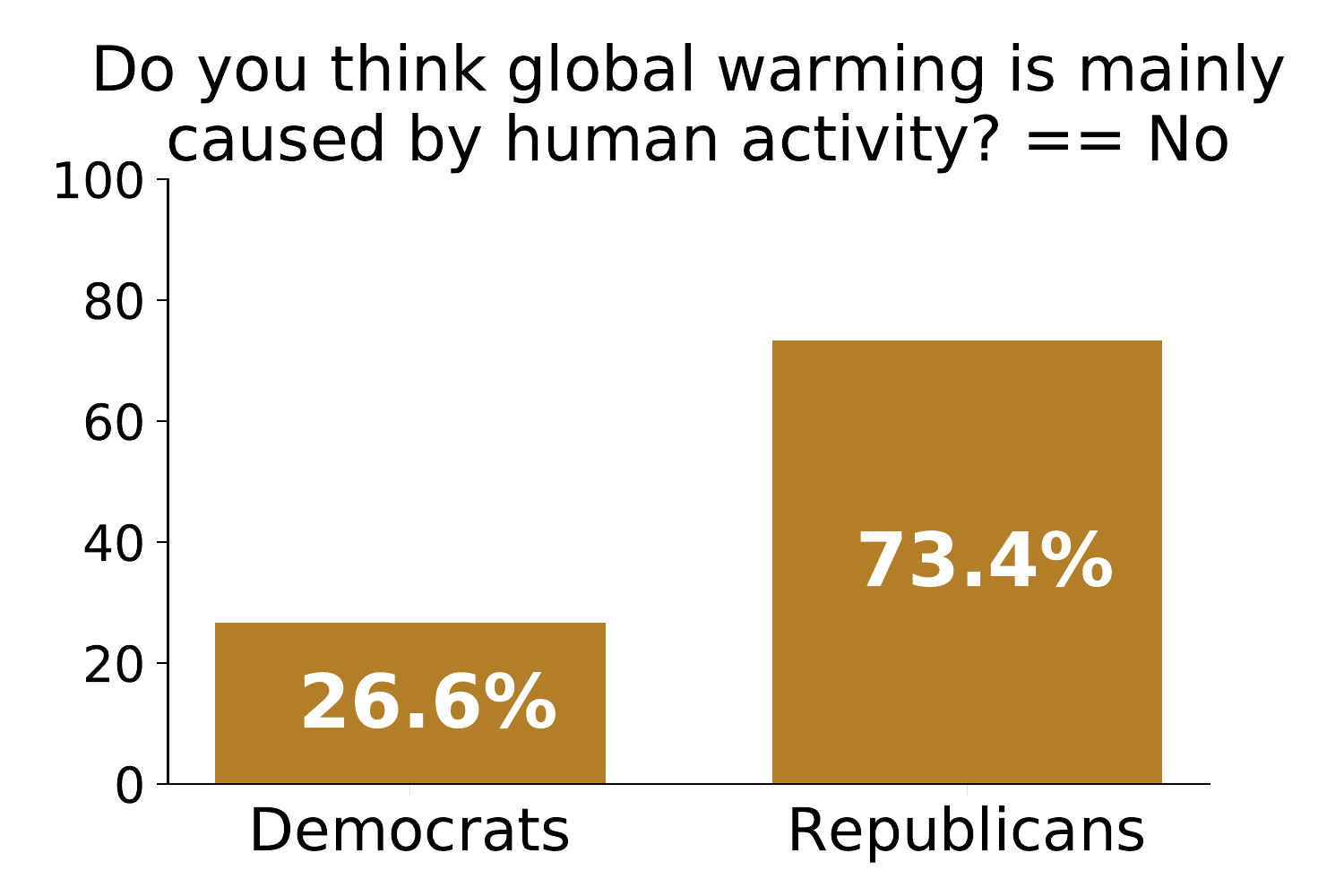}
        \caption{Recommended View 1}
        \label{fig:seedb-mturk-example-target1}
    \end{subfigure}
    \\
    \begin{subfigure}[b]{0.3\textwidth}
        \centering
        \includegraphics[width=\textwidth]{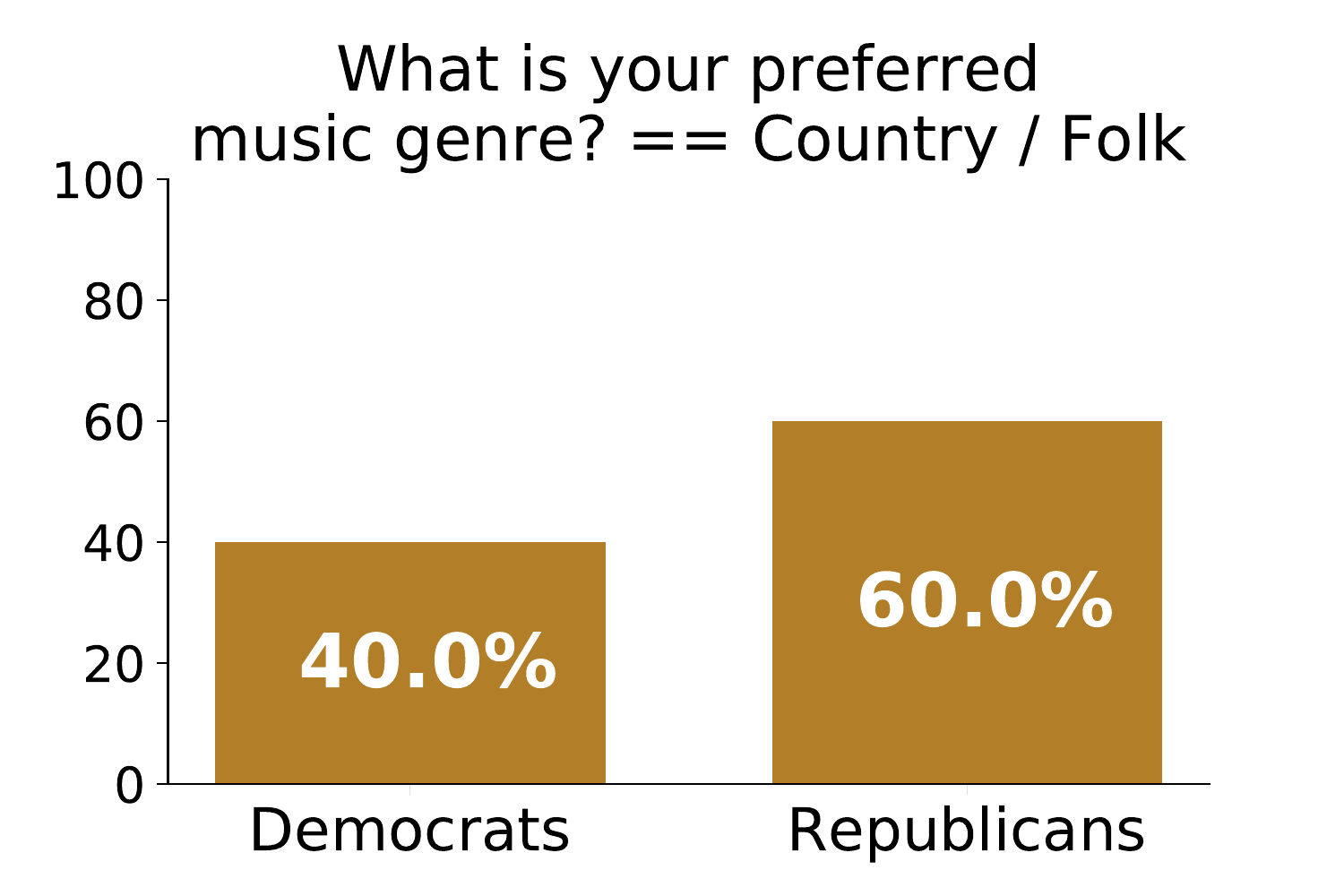}
        \caption{Recommended View 2}
        \label{fig:seedb-mturk-example-target2}
    \end{subfigure}
    \begin{subfigure}[b]{0.3\textwidth}
        \centering
        \includegraphics[width=\textwidth]{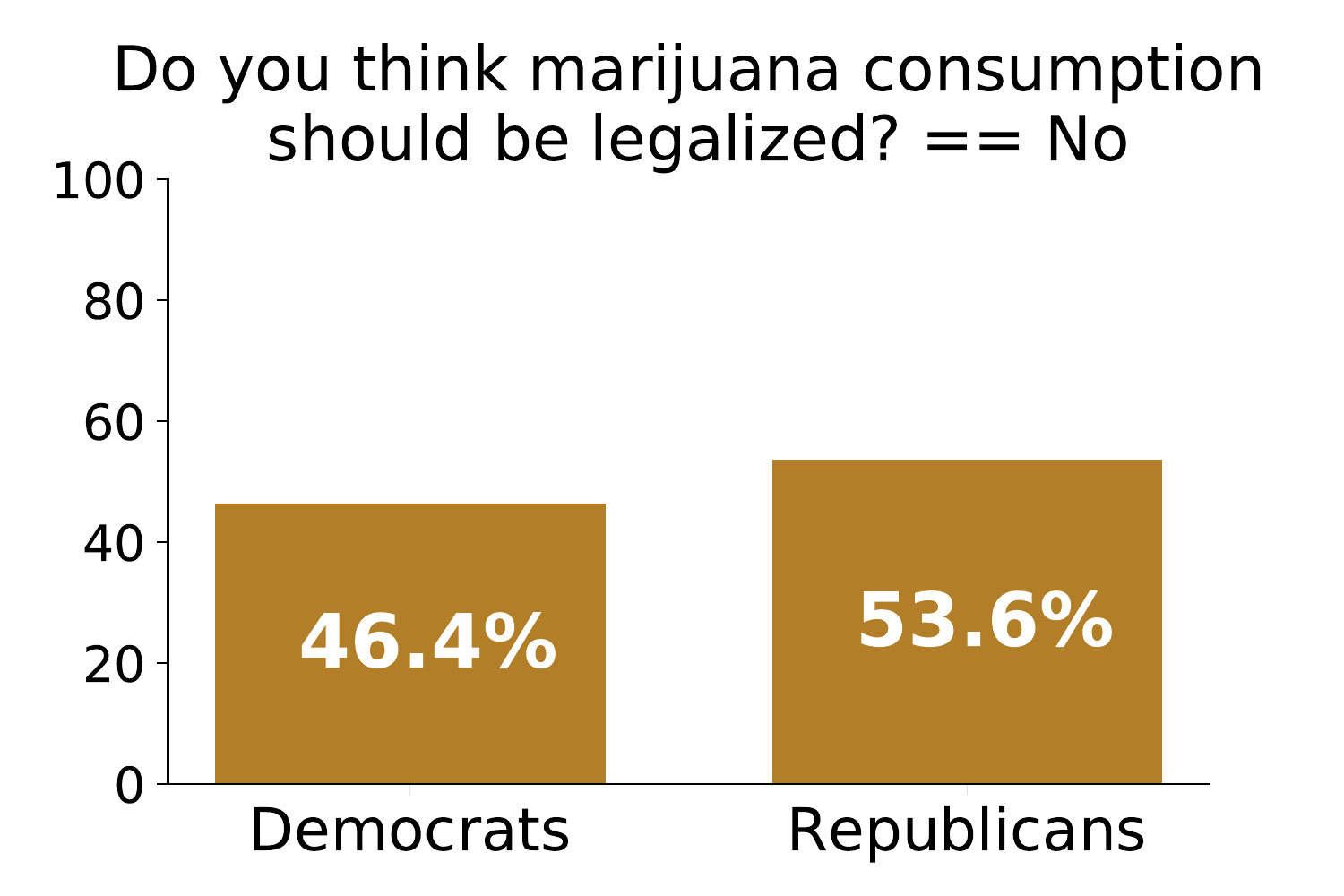}
        \caption{Recommended View 3}
        \label{fig:seedb-mturk-example-target3}
    \end{subfigure}
\end{center}
    \caption{An example of SeeDB \protect\cite{vartak2014seedb} on survey data.}
    \label{fig:seedb-mturk}
\end{figure}
Suppose the user analyzes U.S. voting trends over this data set as shown in Figure~\ref{fig:seedb-mturk-example-refer}. 
Based on this reference view created by the user, the actual SeeDB algorithm over our data set would recommend Figures~\ref{fig:seedb-mturk-example-target1}-\ref{fig:seedb-mturk-example-target3} as visualizations (settings of SeeDB are similar to the ones in \cite{vartak2014seedb}'s Figure~1).
All of the recommended visualizations by SeeDB seem first to be interesting as they clearly indicate a trend reversal and a correlation between certain beliefs and voting behavior. However, these trends could also be solely a random discovery in the way the visualization is just produced because for some reason the dataset by chance just produced such bar graphs for these queries. Indeed, in \ref{sec:anecdotal_examples} we show that some of these recommendations are not \emph{statistically safe recommendations}.

Having a larger dataset may help an analyst to be more certain that the recommendations are actually significant. However, considering the total number of samples used for these visualizations bears the fundamental question, how big the data actually needs to be in order to guarantee that there are no false discoveries. Further, when automatically exploring the dataset at some point with enough filtering as done by SeeDB every dataset becomes ``too small'' to guarantee anything. In the following, we now formalize the recommendation process of SeeDB, introduce the visual recommendation problem and its relation to hypothesis testing.

%% file: new_sec3.tex
\subsection{Problem set-up}\label{sec:model}  

Let $\Omega$ denote the \emph{global sample space}, that is the set comprised of the records of the entire population of interest (e.g., the records of US citizens). We can imagine $\Omega$ in the form of a two-dimensional, relational table with $N$ rows and $m$ columns, where each column represents a \emph{feature} or \emph{attribute} of the records. 


In many practical scenarios the analyst is in possession of a \emph{sample} $\mathcal D$, ${\mathcal D}\subset\Omega$, composed by a much smaller number of records $n<<N$, and aims to \emph{estimate} the properties of $\Omega$ by analyzing the properties of the smaller dataset.

In this work we assume that the input to the recommendation engine is a dataset, $\mathcal D$, consists of $n$ records chosen uniformly at random and independently from the universe $\Omega$ (e.g., our survey data).
We refer to $\mathcal D$ as the \emph{sample dataset} or the \emph{training dataset} and denote by $\mathcal{F}_\mathcal{D}$ the probability distribution that generated the sample $\mathcal D$. Alternatively, one can consider $\mathcal D$ as a sample of size $n$ from a distribution $\mathcal{F}_\mathcal{D}$ that  corresponds to a possibly infinite domain. 

As discussed in \autoref{sec:intro} rather than enabling users to only interpret visualizations for \emph{some particular data set}, we want to make sure that when they try to generalize results a system provides them with statistical guarantees. Hence, we consider the input $\mathcal{D}$  not as a universal but as a \emph{random sample} of the universe, and we focus on observation in $\mathcal{D}$ that apply to the entire universe. For example \emph{prices in city centers are higher than in neighboring districts} instead of \emph{prices in the city center for apartments A, B, C, ... in Manhattan are higher than for apartments X, Y, Z, ... in the Bronx, Queens and Brooklyn neighboring districts for the data collected in 2017 on July 21st}.

We divide the features (or attributes) of the records in $\Omega$, and hence $\mathcal D$, into four groups: 
\begin{enumerate}
    \item \textbf{binary features}, taking values in $\lbrace 0, 1 \rbrace$ (e.g., unit has AC).
    \item \textbf{discrete features with a total order}(e.g.,  \#bedrooms, \#bathrooms).
    \item \textbf{continuous  features with a total order}, (e.g., price of the unit).
    \item \textbf{categorical features}, taking values in a finite unordered  domain (e.g., city, zipcode).
\end{enumerate}
We assume that  each feature is equipped with a \emph{natural metric}. This is achieved, by mapping the values of a feature to a real number (e.g., for a boolean feature by mapping the value $\mathrm{True}$ to $0$, and $\mathrm{False}$ to $1$). 
\subsection{Visualizations} \label{sec:Visualizations}
A common form of visualization used in the dice-and-slice exploration setting of an OLAP (\emph{OnLine Analytical Processing}) data cube is a bar graph. We formalize the type of visualizations we investigate in our approach in the following way:

\begin{definition}
\label{def:visualization}
A visualization $V$ is a tuple 
\[\begin{pmatrix} \mathcal{D}, F, X, Y, \mathrm{AGG}\end{pmatrix}\]
which can be represented as a bar graph and which describes the result of a \emph{query} of the form
\begin{center}
\RecustomVerbatimEnvironment{Verbatim}{BVerbatim}{}
\begin{minted}{sql}
SELECT X, COUNT(Y) FROM D WHERE F GROUP BY X
\end{minted} 
\end{center}
with the aggregate \sqlinline/COUNT/, represented on the $y$ axis,  being partitioned according to the values of a discrete feature \sqlinline/X/, after restricting the records of the input dataset $\mathcal{D}$ being considered to a subset for which the filter predicate \sqlinline/F/ holds.
\end{definition}

The \emph{support} of a visualization $V$ is the number of records of $D$ which satisfy the predicate $F$, and is denoted as $|V|$.
The \emph{selectivity} of a visualization $V$, denoted as $\gamma_{V}$, is defined as the fraction of records which satisfy $F$, that is 
\begin{equation}
    \gamma_{V} := \frac{|\mdataset{}|F|}{|\mdataset{}|}.
\end{equation}
Note that if the group-by attribute $X$ being selected is not discrete, it will also be necessary to determinate a finite set ranges for its value, or ``\emph{buckets}'', to be used in the visualization.\\
The aggregate \sqlinline/COUNT(Y)/ \emph{counts} of records which satisfy the query predicate grouped according to the values of the feature \sqlinline/X/, henceforth referred as the \emph{group-by} feature.

While in this work we focus on the aggregate \sqlinline/COUNT(Y)/, our approach can be extended with minor modifications to the \emph{average}  \sqlinline/AVG(Y)/ aggregate, which is given by the average of the values of the records which satisfy the query predicate grouped according to the values of the group-by feature. 

Though other aggregates like \sqlinline/MIN(Y),MAX(Y),SUM(Y)/ are used in systems like SeeDB \cite{vartak2014seedb} we believe that they do not add value over \sqlinline/COUNT(Y)/ or \sqlinline/AVG(Y)/ aggregates. Even worse, they may lead to misleading visualizations. \sqlinline/MIN(Y),MAX(Y)/ aggregates are \emph{inherently} not suited to represent \emph{statistically significant} behavior of distributions. Rather than using \sqlinline/MIN(Y)/ or \sqlinline/MAX(Y)/ aggregates, a user should consider \emph{conditional expectations} (e.g., aggregates in the form $Y \vert Y > c$ for some constant $c \in \mathbb{R}$) to explore extreme values following the concepts of \emph{extreme-value theory} as described in \cite{gissibl2017big, fasen2014quantifying}. 
While our results can be easily generalized to other types of visualizations (e.g., \emph{heat maps}), for the sake of applicability, in this work we focus on the type captured by Definition~\ref{def:visualization}.
\vspace{-2mm}
\\
\\
\noindent\textbf{Visualizations as distributions:} When considering the \sqlinline/COUNT/ aggregate, it is possible to interpret the data distribution represented by a histogram visualization $V$ as representative of the \emph{probability mass function} (pmf) of the \emph{discrete} random variable $X$ which takes the values of the group by feature \sqlinline/X/, each with probability corresponding  count of the records for each column normalized according to the support of the visualization $V_1$. Such distribution, does indeed correspond to the distribution of the values of the group-by feature  \sqlinline/X/ with respect to the distribution $\mathcal{D}$ after \emph{conditioning} (or \emph{filtering}) with respect to the predicate $F$ associated with $V$. Such correspondence between visualizations and distributions provides us a \emph{natural criteria} to compare visualizations by evaluating their \emph{statistical difference}.
\vspace{-3.5mm}
\subsection{Visualization recommendations}
\label{sec:visualization_recommendation}
Following the paradigm of SeeDB~\cite{vartak2014seedb}, given a first ``\emph{starting} visualization $V_1$ we aim to identify other ``\emph{interesting}'' visualizations to be recommended. In particular, this follows the widespread mantra of \emph{Overview first, zoom and filter, then details-on-demand} \cite{shneiderman1996eyes} where a visualization system ideally lets the user first pick an interesting reference visualization and helps him then to automate the zoom, filter and details-on-demand tasks.

In this work we define a visualization $V_2$ to be interesting with respect to a starting visualization $V_1$ if $V_2$ and $V_1$ are \emph{different}, that is, if they represent a different statistical behavior (i.e., different distribution) of the common group-by feature \sqlinline/X/  under the predicates associated with $V_1$ and $V_2$, respectively. Consistently, the greater the difference between the reference $V_1$ and the candidate $V_2$, the higher the interest of $V_2$ as a \emph{candidate recommendation} for $V_1$.

Note that the constraint according to which possible recommended visualization must share the same group-by feature as the starting visualization is a simple consequence of the fact that we are interested in the study of how different filter predicate conditions do influence the behaviors of the same feature \sqlinline/X/. In the absence of such constraint, the analyst may consider how \emph{different} features behave, thus leading to observations of questionable interest.

Whereas the general problem of recommending interesting visualizations can come either in the way of anomaly detection or as according to the mantra of \cite{shneiderman1996eyes} we want to focus on the latter in this work. That is, we say that a candidate recommendation visualization $V_2$ is interesting with respect to a reference visualization $V_1$ if the two are ``\emph{different enough}''. That is is their distance reaches a threshold $\epsilon$ according to some distance measure~$d$.
\[
\mathcal{V}_2\;\;\text{interesting w.r.t}\;\; \mathcal{V}_1 \Longleftrightarrow d(V_1, V_2) > \epsilon.
\]
In its simplest form, $\epsilon$ may be zero. However it makes more sense to define $\epsilon$ in terms of a minimum visual distance $\epsilon_V$ required by a user to spot a difference\cite{stern2010just} when shown both the reference and a visualization of interest.


While there is a high degree of generality in the selection of the notion of difference between visualizations to be used, in this work, we leverage the correspondence between visualizations which corresponds to Definition~\ref{def:visualization} and the pmf of the group-by feature condition on the filter of the visualization, and we measure the difference between visualizations based on the difference between the associated pmfs. 

In this work, we use the ``\emph{Chebyschev distance}'' to quantify the difference between visualizations. Given two pmfs $\mathcal{D}_{1}$ and $\mathcal{D}_{2}$ over the same support set $\mathcal{X}= \{x_1,x_2,\ldots,x_n\}$, the Chebyscev distance between $\mathcal{D}_{1}$ and $\mathcal{D}_{2}$ is given by:
\begin{equation}
\label{eq:distancedef}
d\left(D_1,D_2\right) = \max_{x\in \mathcal{X}}| \PrD{\mathcal{D}_1}{x} - \PrD{\mathcal{D}_2}{x}|,
\end{equation}
where $\PrD{\mathcal{D}_1}{x}$ (resp., $\PrD{\mathcal{D}_2}{x}$) denotes the probability of a random variable taking value $x$ according to the distribution $\PrD{\mathcal{D}_1}{x}$ (resp., $\PrD{\mathcal{D}_2}{x}$). 

This choice of difference metric is particularly appropriate when comparing visualizations as it highlights the maximum difference between the \emph{relative frequency} of a certain value of the group-by feature in according to the conditional distribution given by the filter predicate of the two visualizations being considered. In other words, when comparing two histogram visualizations the Chebyschev distance captures the maximum difference between pairs of corresponding columns of the histograms.

\section{Statistically safe visualizations and recommendations}\label{sec:statssafvis}

While the statistical pitfalls of exploratory data analysis are well understood and documented (in scholarly papers~\cite{russo2015much, taylor2015statistical},  "surprising statistical" discoveries~\cite{Spurious,pHacking}, and even a famous cartoon~\cite{cartoon}), the connection to visualizations, and specifically visual recommendations engines, have only recently begun to be rigorously studied~\cite{controlfd, kim2015rapid, Chaudhuri:1998:RSH:276304.276343}.
In this section, we formulate the statistical problem in the visual recommendation context and explore simple probabilistic techniques to solve it. A more powerful method, based on statistical learning theory is presented in the Section~\ref{sec:VC}.  While, for concreteness, in our presentation we focus on the SeeDB paradigms, the proposed model and techniques can  be used for other recommendation systems too (e.g., Data Polygamy~\cite{DataPolygamy}).

A first crucial observation is that a system that provides visual representation of data and aims to highlight an \emph{interesting relationship} between visualizations 
should provide tools to allow the analyst to ascertain that the phenomena being observed are actually \emph{statistically relevant}. 
A recommender system should ensure that visualized result displays characteristics that are \emph{non-random} and \emph{visually intelligible}. That is a user looking at two visualizations should both be able to understand that they are different and why they are different without worrying whether visual features are due to missing support or random noise.


Recall that a dataset \dataset{} available to our visualization decision algorithm is a sample from an underlying distribution $\mathcal{F}_\mathcal{D}$. 
Assume that a particular visualization is interesting with respect to $\mathcal{D}$. The question we are trying to answer is, how likely it is that the visualization is also interesting with respect to the underlying distribution $\mathcal{F}_\mathcal{D}$.

In probabilistic terms, each query with a filter predicate \sqlinline/F/ corresponds to an event $E$ over $\Omega$. 
For concreteness consider a visualization of a histogram of a (discrete and finite) variable $X$  conditioned on an event $E \not= \emptyset$, denoted $X \vert E$.

 The true values (in $\mathcal{F}_\mathcal{D}$) for $k \in \dom{X}$ are given by
\[p_k := \mathbb{P}(X = k \vert E) = 
\frac{\mathbb{P}\left( X = k, E\right)}{\mathbb{P}(E)}.
\] 
We estimate these values in a dataset $\{X_1,\dots,X_n\} \subseteq \mathcal{D}$ of size $n$ by
\begin{equation}\label{eq:empiricalp}
  \hat{p}_k :=  \frac{\sum_{i=1}^n \mathbbm{1}_{\lbrace X_i = k, \; X_i \in E \rbrace}}
{\sum_{i=1}^n \mathbbm{1}_{\lbrace X_i \in E \rbrace}}.  
\end{equation}

If in $\mathcal{D}$ the histogram of $X\vert E$ is visually different from the histogram of $X$, what can we rigorously predict about the difference between the histograms in $\mathcal{F}_\mathcal{D}$?
 
Therefore, we say that the difference between two visualizations $\mathcal{V}_1$ and $\mathcal{V}_2$  is \emph{statistically significant} if and only if the difference observed between the two in the finite sample $\mathcal{D}$ is due to an actual difference between the two histograms with respect to the distribution $\mathcal{F}_{\mathcal{D}}$.The recommendation problem thus becomes in its general form to recommend a candidate visualization $\mathcal{V}_2$ for a reference $\mathcal{V}_1$  only if their corresponding histograms are statistically different with respect to the \emph{true} underlying distribution $\mathcal{F}_\mathcal{D}$.

Our goal is to verify that interesting visualization flagged by our algorithm with respect to $\mathcal{D}$ \emph{generalize} to interesting visualizations with respect to  $\mathcal{F}_\mathcal{D}$.

%% file: sec4.tex
\subsection{Classical statistical testing}
In the \emph{classical} statistical testing setting, our problem could be framed in two ways: Either it could be formulated as a \emph{goodness-of-fit} test or as a \emph{homogeneity} test. In a \emph{goodness-of-fit}-test for a given starting reference visualization $\mathcal{V}_1$ and a candidate visualization recommendation  $\mathcal{V}_2$ a \emph{hypothesis} is considered in the form of whether certain statistical attributes of the reference visualization (i.e. the expected attributes) fit the corresponding observed attributes from the candidate query. Classical tests include the single test $\chi^2$-test for discrete distributions or the Kolmogorov-Smirnov test for continuous random variables. In the visualization context thus a candidate query would be selected as interesting when the null hypothesis of the attributes being similar is rejected.
A \emph{homogeneity}-test on the other hand tests the hypothesis that two samples were generated by the same underlying distribution, which may be unknown. This is done, for example, by a two samples $\chi^2$-test, or a $t$-test comparing one column in two histograms. A system based on \emph{homogeneity}-tests would then select a candidate query as interesting iff the sample corresponding to $\vis{2}$ is not homogeneous with the underlying data distribution of the reference query $\vis{1}$.

However, there are major difficulties in applying standard statistical tests to the visualization problem.
 First, depending on the input data the correct test needs to be selected. For example, when using a $\chi^2$-test over discrete attributes, each bucket must not be empty. A general rule of thumb to make sure estimates are reliable is to have at least $5$ samples per bucket. Further, there should be enough samples to actually use the $\chi^2$-test. Else, Fisher's exact test should be used for small sample sizes. In addition to each test being only applicable to certain input data, they all do guarantee a \emph{different} notion of interest. A user that is presented with the test results of one or multiple tests usually will not be able to immediately connect the results to the notion of a \emph{significant visual difference} as described in \autoref{sec:visualization_recommendation}. This brings up the problem on how comparable results of e.g a $t$-test against a $\chi^2$-test actually are in terms of \emph{visual difference}.

Second, when blindly throwing statistical tests at the visualization problem to deal with different types of input data and different sample sizes queries return, the question is whether the hypothesis being tested are not too \emph{simple} for recognizing visual difference in a meaningful way. Consider for this a $t$-test that essentially compares whether the observed mean resembled the expected mean. Naturally, a consequence is that if they differ the candidate query should get recommended. This may however lead to many wrong recommendations merely because the null hypothesis used is too \emph{simple} and gets rejected too often. The solution could be to use a test better suited for the problem, e.g.  in the form of a $\chi^2$-test. However, as we show in \autoref{sec:exp_chisq} a $\chi^2$-test is not suited best to spot a notion of \emph{statistical significant visual difference} and comes with its own problems as pointed out in \cite{delucchi1983use}. There is no free lunch and thus no universal single test that solves the visualization recommendation in general.

Third, most tests only offer merely \emph{asymptotic guarantees} because of the test statistic they use. Especially for skewed distributions or queries that return only a small number of rows this is problematic. Consider once again a $\chi^2$-test and a heavily skewed distribution over e.g. $20$ buckets. It is then very unlikely that the test statistic for the number of samples used in a visualization setting is already $\chi^2$-distributed.

Fourth, recommendations based on tests are not necessarily \emph{symmetric} in the sense that if the candidate query was used as reference query, the old reference query would not get necessarily recommended at all. This is especially true for the $\chi^2$-test.

Lastly, one might be tempted to simply combine statistical testing with a magical cutoff or subsequent selection of visualizations based on the distance measure introduced in \autoref{eq:distancedef}. I.e. an algorithm could be to first apply statistical testing to get a candidate set of potentially interesting visualizations, rank them after the distance measure and then select all visualizations as interesting that have a distance higher than $\epsilon_S$ with respect to the reference visualization. However, this approach merely delays the problem of potentially making false discoveries: Though according to the tests the visualizations may be indeed \emph{different} according to the difference guarantees the employed tests offer, they do not necessarily need to be \emph{significantly different} enough. This again can be shown using the example in \autoref{sec:exp_chisq}.
\vspace{-3mm}
\subsection{Recommendation validation via estimation}
The goal of \texttt{VizRec} is to provide an efficient and rigorous way to verify that two visualizations are indeed \emph{statistically different} with ``\emph{finite-sample}'' guarantees.

In VizRec we use the sample dataset \dataset{}, and the visualizations obtained from it, in order to obtain \emph{approximations} of the visualization according to the entire global sample space $\Omega$ using the Chernoff bound and later VC dimensions. 

Consider a single histogram visualization $\vis{1}$, and assume it is comprised of $K$ bars, one for each of the $K$ possible values of the chosen group-by feature $X$. Let $p_{\vis{1}}(x_1),\ldots,p_{\vis{1}}(x_k)$,  denote the normalized bars corresponding to $\vis{1}$. Note that such bars denote the probability of a randomly chosen record from $\Omega$ being such that $X=x_i$ conditioned on the fact that such record satisfies the predicate associated with $\vis{1}$.

Using the sample set \dataset{} we compute an approximation for the $p_{\vis{1}}(x_i)$, which we denote as $\hat{p}_{\vis{1}}(x_i)$,  following~\eqref{eq:empiricalp}. As only these approximations can be computed from the available data, any choice regarding which visualizations should be recommended may depend only on such approximations.

In order to argue guarantees regarding the reliability of such decisions, it is necessary to bound the maximum difference between the correct and estimated sizes of bars in the normalized histograms.

In particular for a given $\delta$ (i.e., our level of control for \emph{false positive recommendations}) we want to compute the minimum value $\epsilon \in (0,1)$
$$
\PrD{\mathcal{\mdataset{}}}{|p_{\vis{1}}(x_i) - \hat{p}_{\vis{1}}(x_i)| > \epsilon}<\delta.
$$

Such value $\epsilon$ would in turn quantify the accuracy of the estimation of the $p_{\vis{1}}(x_i)$'s obtained by means of their empirical counterpart $\hat{p}_{\vis{1}}(x_i)$'s.

Let $F$ denote the predicate associated with our visualization $\vis{1}$. We denote as $\Omega|F$ (resp., $\mdataset{}|F$) the subset of $\Omega$ (resp., \dataset{}) which is composed by those records that satisfy the predicate $F$. Given a choice of group-by attribute $X$ the value $p_{\vis{1}}(x_i)$ (resp., $\hat{p}_{\vis{1}}(x_i)$) corresponds to (resp., is computed as) the \emph{relative frequency} of records such that $X=x_i$ in $\Omega|F$ (resp., $\mdataset{}|F$).

\begin{fact}\label{fact:uniformsubsample}
Let \dataset{} be an uniform random sample of $\Omega$ composed by $m$ records. For any choice of predicate, as specified in Definition~\ref{def:visualization}, the subset $\mdataset{}|F$ is a uniform random sample of 	$\Omega|F$ of size $|\mdataset{}|F|$.
\end{fact}

Fact~\ref{fact:uniformsubsample} is a straightforward consequence of the fact that \dataset{} is an uniform sample of $\Omega$.

From Fact~\ref{fact:uniformsubsample} and from the definitions of the $p_{\vis{1}}(x_i)$'s and of the $\hat{p}_{\vis{1}}(x_i)$'s, clearly follows that the $\hat{p}_{\vis{1}}(x_i)$'s are \emph{unbiased estimators} for the $p_{\vis{1}}(x_i)$'s. That is for, every value of the group-by feature we have:
\begin{equation}
    \Expe{\mdataset{}}{\hat{p}_{\vis{1}}(x_i)} = p_{\vis{1}}(x_i).
\end{equation}

In order to bound the estimation error $|p_{\vis{1}}(x_i) - \hat{p}_{\vis{1}}(x_i)|$ is therefore sufficient to bound the \emph{deviation from expectation} (i.e., $p_{\vis{1}}(x_i)$) of the empirical estimate (i.e., $\hat{p}_{\vis{1}}(x_i)|$).

Chernoff-Bounds~\cite{Mitzenmacher:2005:PCR:1076315}, allows to obtain such bounds as:
\begin{align}\label{eq:chernoff1}
\PrD{\mathcal{\mdataset{}}}{|p_{\vis{1}}(x_i) - \hat{p}_{\vis{1}}(x_i)|>\epsilon} &\leq e^{-2|\mdataset{}|F|\epsilon^2}
\end{align}


Recall our definition of \emph{selectivity of a visualization} as $\gamma_{\vis{1}} = \frac{|\mdataset{}|F|}{\mdataset{}}$, we can then rewrite~\eqref{eq:chernoff1} as:
\begin{align}\label{eq:chernoff2}
\PrD{\mathcal{\mdataset{}}}{|p_{\vis{1}}(x_i) - \hat{p}_{\vis{1}}(x_i)|>\epsilon} &\leq e^{-2\gamma_{\vis{1}}n\epsilon^2},
\end{align}

where $n= |\mdataset{}|$. A clear consequence of~\eqref{eq:chernoff2}, is that the higher (resp., the lower) the selectivity of a visualization, the  higher (resp., the lower) the quality of the estimate. 

In \cite{kim2015rapid}, the authors use the same kind of bounds to develop a sampling algorithm which ensures  the visual property of \emph{relative ordering}. That is, for any pairs of vars corresponding to two possible values $x_1$ and $x_2$ of the same group-by feature $X$, if $\hat{p}_{\vis{1}}(x_1)>\hat{p}_{\vis{1}}(x_2)$, then, with high probability,  $p_{\vis{1}}(x_1)>p_{\vis{1}}(x_2)$ holds as well.

While the method previously described based on an application of the Chernoff bound appears to be very useful and practical, it is important to remark that in a single application it may only offer guarantees on the quality of the approximation of \emph{one} bar from a \emph{single visualization}. 

While it is in general possible to combine multiple applications of the Chernoff bound, the required correction leads to a quick and marked decrease of the quality of the bound. As an example if our visualization $\vis{1}$ is composed by $K$ bars, a bound on the quality of the approximation of all of the bars will result in:
\begin{align*}
\PrD{\mathcal{\mdataset{}}}{\max_{i=1,\ldots,K}|p_{\vis{1}}(x_i) - \hat{p}_{\vis{1}}(x_i)|>\epsilon} &\leq Ke^{-2\gamma_{\vis{1}}n\epsilon^2}.
\end{align*}

This bound is obtained using the \emph{union bound}~\cite{Mitzenmacher:2005:PCR:1076315}. While potentially tolerable for a small value of $K$, for large values the performance decrease can possibly lead to a complete loss of significance of the bound itself.

\subsection{Correcting for Adaptive Multi-Comparisons}
\label{sec:multicomparison}
The previous section showed how we can evaluate the interest of a single visualization as a statistical test.
However, this does not yet control the multiple comparisons problem which is inherent in the  recommendation system process that explores many possible visualizations. Clearly, if we let a recommendation system explore an unlimited number of possible visualizations, it will eventually find an ``\emph{interesting}'' one, even in random data. How do distinguish real interesting visualization from spurious ones that are the results of random fluctuation in the sample?

The simplest and safest method to avoid the problem is to test every visualization recommendation on an independent sample not used during the exploration process that led to this recommendation. While easy and safe, this method is clearly not practical for a process that explores many possible visualizations. Data is limited, and we cannot set aside a holdout sample for each possible test. The process needs access to as much data as possible in the exploration process in order to discover all interesting insights. Can we control the generalization error when computing a number of visualizations based on one set of data?

Assume that in our exploration of possible interesting visualizations we tried $\ell$ different visualization patterns, and we computed for each of these patterns a bound $h_i$, $i=1,\dots,\ell$, on the probability that the corresponding observation in the sample \dataset{} does not generalize to the distribution with respect to the entire global sample space $\Omega$.
It is tempting to conclude that the probability that any of the $\ell$ visualizations does not generalized is bounded by $\sum_{i=1}^\ell h_i$. 
Unfortunately, this probability is actually much larger when the choice of the tested visualization depends of the outcome of prior tests. 

This phenomenon is often referred to as Freedman's paradox~\cite{} and the only known practical approach to correct for it is to sum the error probability of all possible tests, not only the tests actually executed~\footnote{Theoretical methods, such as differential privacy~\cite{dwork2015reusable} claim to offer an alternative method to address this issue. 
 In practice however, the signal is lost in the added randomization before it becomes practical.}. Note that  standard statistical techniques for controlling the Family-Wise- Error-Rate (FWER) or the False Discovery Rate (FDR) require that the collection of tests is fixed independent of the data and therefore do not apply to the adaptive exploration scenario.





In the visualization setting, we could decide a-priory that we only consider visualizations of a particular set of patterns, say conditioning on no more than $k$ features. Such restriction defines a bound on the total size of the search space, say $M$. If we now explore the search space and recommend visualization that pass the individual visualization test with confidence level $\leq \alpha/M$ we are guaranteed that the probability that any of our recommendations does not generalize is bounded by $\alpha$. As we show in the experiments section this method is only effective for relatively small search space. Next, we present a novel technique that is significantly more powerful in large and more complex search spaces.

%% file: vc.tex
\section{Statistical Guarantees Via Uniform Convergence Bounds}
\label{sec:VC}
We now present a powerful alternative approach, providing strong and practical statistical guarantees through a novel application of VC-dimension theory. 

VC-dimension is usually considered a highly theoretical concept, with limited practical applications, mostly because of the difficulty in estimating the value of the VC-dimension of interesting learning concept class.

Surprisingly, we develop a simple and effective method to compute the VC-dimension of a class of predicate queries which are used to generate the visualizations according to Definition~\ref{def:visualization}, which leads to a practical and efficient solution to our problem. We start with a brief overview of the VC theory and its application to sample complexity. We then discuss its specific application to our visualization problem, emphasizing a simple, efficient, and easy to compute reduction of the theory to practical applications.

\subsection{VC dimension}
\label{sec:vcapproach}
The Vapnik-Chernovenkis (VC) dimension is a measure of the
complexity or expressiveness of a family of indicator functions (or equivalently
a family of subsets)~\cite{vapnik2015uniform}. 
Formally, VC-dimension is defined on {\em range spaces}:

\begin{definition}\label{defn:rangespace}
 A {\em range space} is a pair $(X,R)$ where $X$ is a (finite or infinite) set
 and $R$ is a (finite or infinite) family of subsets of $X$. The members of $X$
 are called {\em points} and those of $R$ are called {\em ranges}.
 \end{definition}
Note that both $X$ and $R$ can be infinite. Consider now a  projection of the ranges into a finite set of points $A$:
\begin{definition}\label{defn:proj}
Let $(X,R)$ be a range space and let $A\subset X$ be a finite set of points in $X$. 
\begin{enumerate}
    \item 
The {\em projection} of $R$ on
$A$ is defined as $$P_R(A)=\{r\cap A ~:~ r\in R\}.$$


\item If $P_R(A)=2^{|A|}$, then $A$ is said
to be {\em shattered by $R$}.
\end{enumerate}
\end{definition}

\def\VC{\mathsf{VC}}
The VC-dimension of a range space is the cardinality of the largest set
shattered by the space:
\begin{definition}\label{defn:VCdim}
  Let $(X,R)$ be a range space. The
  {\em VC-dimension} of $(X,R)$, 
  denoted $\VC(X,R)$ 
  is the maximum cardinality of
  a shattered subset of $X$. If there are arbitrary large shattered subsets,
  then $\VC(X,R)=\infty$.
\end{definition}

Note that a range space $(X,R)$ with an arbitrarily large (or infinite) set of points $X$ and
an arbitrary large family of ranges $R$ can have bounded  VC-dimension (see section~\ref{ConcreteVC}). 

A simple example is the family of closed intervals in $\mathbb{R}$, where $X = \mathbb{R}$, and $R$ corresponds to the (infinite) set of all possible closed intervals intervals $[a,b]$, such that  $\forall a,b \in \mathbb{R}$ we have $0\leq a\leq b\leq 1$). Let $A=\{x,y,z\}$ be \emph{any} subset of $\mathbb{R}$ such that $x\leq y\leq z$. No interval in $R$ can define the subset $\{x,z\}$ so the VC-dimension of this range space is $< 3$.  This observation is generalized in the well known result:
\begin{lemma}\label{lem:intervalunions}
   The VC-Dimension of the union of $d$ closed intervals in $\mathbb{R}$  equals $2d$.
 \end{lemma}

VC-dimension, allows to characterize the sample complexity of a learning problem, that is it allows to obtain a tradeoff between the number of sample points being observed by a learning algorithm and the performances achievable by the algorithm itself. 

Consider a range space $(X,R)$, and a fixed range $r\in R$. If we sample uniformly at random a set $S\subset X$ of size $m := |S|$ we know that the fraction $\frac{|S\cap r|}{|S|}$
rapidly converges to the frequency of elements of $r$ in $X$. Furthermore, there are standard bounds (Chernoff, Hoeffding~\cite{}) for evaluating the quality of this approximation. The question becomes much harder when we want to estimate simultaneously the sizes or frequencies of all ranges in $R$ using one sample of $m$ elements.
A finite VC-dimension implies an explicit upper bound on the number of random samples needed to achieve that within pre-defined error bounds (the \emph{uniform convergence property}). 

For a formal definition we need to distinguish between finite $X$, where we case estimate the sizes $r$, and infinite $X$, where we estimate $Pr(r)$, the frequency of $r$ in a uniform distribution over $X$.
\begin{definition}[Absolute approximation]
Let $({X},R)$ be a range space and let $0\leq \epsilon \leq 1$.  A subset $S\subset X$ is an absolute $\epsilon$-approximation for ${X}$ iff for all $r\in R$ we have that for finite $S\subseteq X$,
\begin{equation}
\left| \frac{|r|}{|X|} - \frac{|S\cap r|}{|S|}\right| \leq \epsilon.
\end{equation}
\end{definition}

In~\cite{har2011relative} show an interesting connection between the VC dimension of a range space $({X},R)$ and the number of samples which are necessaries in order to obtain absolute $\epsilon$-approximations of ${X}$ itself 

\begin{theorem}[Sample complexity~\cite{har2011relative}]\label{thm:vcsamplec}
Let $(X,R)$ be a range-space of VC-dimension at most $d$, and let and $0<\epsilon,\delta<1$. Then, there exists an absolute positive constant $c$ such that any random subset $S\subseteq X$ of cardinality 
\begin{equation}
|S| \geq \frac{c}{\epsilon^2}\left(d + \log_2 \delta^{-1}\right)
\end{equation}
is an $\epsilon$-approximation for $X$ with probability at least $1-\delta$.
\end{theorem}
The constant $c$ was shown experimentally~\cite{c2loeffler} to be at most $0.5$\footnote{Indeed, we use $c=0.5$ in our experimental evaluation.}.





\subsection{Statistically Valid Visualization through VC dimension}\label{ConcreteVC}

To apply the uniform convergence method via VC dimension to the visualization setup, we consider a range space $(\Omega,R)$, where $\Omega$ is global domain, and $R$ consists of all the possible subsets of $X$ that can be selected by visualizations predicates. That is, $R$ includes all the subsets that correspond to \emph{any} bar for \emph{any} visualization which can be selected using the appropriate predicate filter. Given a choice of possible allowed predicates, we refer to the associate set of ranges as the ``\emph{query range space}'' and we denote it $Q$.

The VC dimension of a query range class is a function of the type of select operators (i.e., $>,<,\geq,\leq,=,\neq$) and
the number of (non-redundant) operators allowed on each feature in the construction of the allowed predicates. Note that depending on the \emph{domain} of the selected features and the complexity according to which the predicate filters can be constructed, the number of possible predicates may be infinite. In order to use the VC-approach it is however sufficient to efficiently compute a finite \emph{upper bound} of the VC-dimension of the set of \emph{allowed predicates}.
We discuss an efficient method for bounding the VC-dimension of a query range space in the next subsection (See Subsection~\ref{ConcreteVC}).

In order to deploy the general results from the previous section, we have to verify that the sample \dataset{} provides an $\epsilon$-approximation for the values $p_{\mathcal{V}}$ for all the visualizations considered in the query range space $Q$.

To this end, it is useful to introduce the following, well known, property of VC dimension:

\begin{fact}
Let $(X,R)$ be a range space of VC dimension $d$. For any $X'\subseteq X$, the VC-dimension of $(X',R)$ is bounded by $d$.
\end{fact}

Using this fact in conjunction with Theorem~\ref{thm:vcsamplec} we have:

\begin{lemma}\label{lem:epsilonbar}
Let $(\Omega, Q)$ denote the range space of the queries being considered with VC dimension  bounded by $d$, and let $\delta\in (0,1)$. Let \dataset{} be a random subset of $\Omega$. Then there exists a constant $c$, such that with probability at least $1-\delta$ for any filter $F$ defined in  $Q$ we have that the subset $\mdataset|F$ is an $\epsilon_{F}$-approximation of $\Omega|F$ with:
\begin{equation*}
    \epsilon_F \geq \sqrt{\frac{c}{|D|F|} \left( d+\log_2 \delta^{-1}\right)}.
\end{equation*}
\end{lemma}
\begin{proof}
Fact~\ref{fact:uniformsubsample} ensures that given the dataset \dataset, for any choice of a predicate $F$ we have that $\mdataset{}|F$ is a random sample of $\Omega|F$. 
Therefore regardless of the specific choice of the predicate, we have that the VC dimension of the reduced range $(\Omega|F,Q)$  is bounded by $d$. From Theorem~\ref{thm:vcsamplec} we have that if:
\begin{equation}
|\mdataset{}| \geq \frac{c}{\bar{\epsilon}}\left(d+\log_2 \delta^{-1}\right) 
\end{equation}
then $\mdataset{}|F$ is an $\bar{\epsilon}$ approximation for the respective set $\Omega|F$.
\end{proof}

Lemma~\ref{lem:epsilonbar} provides us an efficient tool to evaluate the quality of our estimations $ \hat{p}_{\mathcal{V}}$ of the actual ground truth values $ p_{\mathcal{V}}$ for any choice of predicate associated with the visualization. 
In particular, Lemma~\ref{lem:epsilonbar} verifies that the quality decreases gradually the more \emph{selective} the predicate associated with a visualization is. That is, the smaller the cardinality of $|\mdataset{}|F|$, the higher the \emph{uncertainty} $\bar{epsilon}$ of the estimate is.

\begin{corollary}
Let \dataset{} be a random sample from $\Omega$, and let $Q$ be a query range space with VC dimension bounded from above by $d$. For any visualization with  $\mathcal{V} \in Q$ and for any value $\delta\in(0,1)$ we have that
\begin{equation}
\Pr\{\left| p_{\mathcal{V}}(X=x_i) - \hat{p}_{\mathcal{V}(X=X_i)}\right| \geq  \bar{\epsilon}\}< \delta,
\end{equation}
where
\begin{equation}
\bar{\epsilon} \geq \frac{c}{|D|F|}\left(d+\log_2 \delta^-1\right),
\end{equation}
$F$ denotes the predicate associated with the visualization $\mathcal{V}$ and and $X$ denotes the group-by feature being considered. 
\end{corollary}

\subsection{The VizRec recommendation validation criteria}

Consider now a given reference visualization $\vis{1}$ and a candidate recommendation $\vis{2}$, both using $X$ as the group-by feature, were the the domain of $X$ has $K$ values (i.e., $dom(X)=\{x_1,\ldots,x_K\}$). From Lemma~\ref{lem:epsilonbar}, we have that with probability $1-\delta$ the empirical estimates of the normalized columns are accurate within $\bar{\epsilon}$, which depends on the size of the subset of the sample dataset \dataset{} used for the reconstruction of $\vis{2}$.

As argued in Section~\ref{sec:statssafvis}, we consider a candidate visualization worth of being recommended if it represents a different statistical behavior of the group-by feature with respect to the reference query. 

Our VizRec strategy operated by comparing the values $\hat{p}_{\vis{1}}(x_i)$ and $\hat{p}_{\vis{2}}{x_i}$ for all values $x_i$ in the domain of the chosen group-by feature. 
Let $\bar{\epsilon}_1$ (resp., $\bar{\epsilon}_2$) denote the uncertainty such that with probability at least $1-\delta$ we have $\|p_{\vis{1}}(x_i) - \hat{p}_{\vis{1}}(x_i)\|\leq \bar{\epsilon}_1$ and $\|p_{\vis{2}}(x_i) - \hat{p}_{\vis{2}}(x_i)\|\leq \bar{\epsilon}_2$ accoding to Lemma~\ref{lem:epsilonbar}. 
If it is the case that $|\hat{p}_{\vis{1}}(x_i)-\hat{p}_{\vis{2}}(x_i)|> \bar{\epsilon_1}+ \bar{\epsilon_2}$ then we can conclude that with probability at least $1-\delta$ we have $p_{\vis{1}}(x_i)\neq p_{\vis{2}}(x_i)$. 

That is VizRec recognizes as \emph{statistically different} (and hence, interesting) only pairs of visualizations for which the most different pair of corresponding columns differs by more than the error in the estimations from the sample. If that is the case, it is possible to guarantee that $\vis{1}$ and $\vis{2}$ are indeed according to the Chebyschev measure.
Due to the uniform convergence bound ensured by the application of  VC dimension, we can have that the probabilistic guarantees of this control hold \emph{simultaneously} for all possible pairs of reference and candidate recommendation visualizations. The advantage of this approach compared to the use of multiple Chernoff bounds is discussed in Appendix~\ref{app:vcchernoff}. Further, our VC dimension approach is \emph{agnostic} to the adaptive nature of the testing as it accounts \emph{preventively} for \emph{all} possible evaluations of pairs of visualizations. Threrefore, we have:

\begin{theorem}\label{thm:uniformFWERcontrol}
For any given $\delta \in (0,1)$, VizRec ensures FWER control at level $\delta$ while offering visual recommendations.
\end{theorem}

\begin{algorithm}
\caption{VizRec: Visual Recommendations with VC dimension}\label{alg:vizrec}
\begin{algorithmic}[1]
\Procedure{VizRec}{}
\Statex \textbf{Input:} Starting visualization $\vis{1}$ , query space $Q$, sample dataset \dataset{}, FWER target control level $\delta\in (0,1)$.
\Statex \textbf{Output:} A set of $Y$ statistically safe recommendations sorted according to decreasing interest.
\State $Y\leftarrow []$ \Comment{Empty list of recommendations}
\State $X\leftarrow$ the group-by feature being considered.
\State $F_{\vis{1}}\leftarrow$ the predicate associated with $\vis{1}$.
\State $\bar{\epsilon}_1\leftarrow \frac{d + \log_2\delta^{-1}}{2|D|F_{\vis{1}}|}$ \Comment{Uncertainty in $\vis{1}$ approx.}
\For{ \textbf{all} $\mathcal{V}'\in Q$}
\State $F_{\mathcal{V}'}\leftarrow$ the predicate associated with $\mathcal{V}'$.
\State $\bar{\epsilon}'\leftarrow \frac{d + \log_2\delta^{-1}}{2|D|F_{\mathcal{V}'}|}$ 
\State $dist \leftarrow  \max_{x_i\in dom(X)}|\hat{p}_{\vis{1}}(x_i)-\hat{p}_{\vis{2}}(x_i)|$
\State $interest \leftarrow dist - \left( \bar{\epsilon}_1+ \bar{\epsilon}_2 \right)$
\If {$dist \geq \bar{\epsilon}_1+ \bar{\epsilon}_2$}
\State $Y.append([\mathcal{V}', interest]$
\EndIf
\State \textbf{or} if stricter criteria with $\epsilon_V$
\If {$dist \geq \max \lbrace \bar{\epsilon}_1+ \bar{\epsilon}_2, \epsilon_V\rbrace$}
\State $Y.append([\mathcal{V}', interest]$
\EndIf
\EndFor
\State \textbf{return} sort $Y$ according to the interest value.
\EndProcedure
\end{algorithmic}
\end{algorithm}

This criteria can be strengthened by imposing a higher threshold of difference between two visualization in order for a candidate visualization to be considered interesting. As an example, in Section 3.2, we discuss the possible use of a threshold $\epsilon_{V}$ denoting visual discernability. When using this, more restrictive constraint, VizRec would accept a candidate visualization as interesting only if 
\begin{equation*}
\max_{ x_i\in dom(X)} |\hat{p}_{\vis{1}}(x_i) -\hat{p}_{\vis{2}}(x_i) |> \max\lbrace \bar{\epsilon}_1+\bar{\epsilon_2}, \epsilon_{\mathcal{V}}\rbrace.
\end{equation*}
We present a simplified psudocode of our VizRec procedure  in Algorithm~\ref{alg:vizrec}.

Our VizRec approach operates as the equivalent of a two-sample test, in the sense that we assume that in general there is uncertainty in the reconstruction of both the reference visualization $\vis{1}$ and of the candidate $\vis{2}$. In some scenarios, it may be possible to assume that the reference visualization is given \emph{as exact}. In such case in order for the candidate $\vis{2}$ to be recommended it would be sufficient that 
\begin{equation*}
\max_{ x_i\in dom(X)} |\hat{p}_{\vis{1}}(x_i) -\hat{p}_{\vis{2}}(x_i) |> \bar{\epsilon_2}.
\end{equation*}

After identifying as set of recommendations whose interest is guaranteed with probability at least $1-\delta$, VizRec ranks them according to the difference between their ``\emph{empirical interest}'' (i.e., $\max_{ x_i\in dom(X)} |\hat{p}_{\vis{1}}(x_i) -\hat{p}_{\vis{2}}(x_i) |$) and the uncertainty of the evaluation of such measure (i.e., $ \bar{\epsilon_1}+ \bar{\epsilon_2}$). While somewhat arbitrary, we chose this heuristic as it allows to emphasize the intrinsic value of visualizations with large support over those with small support.


\subsection{The VC dimension of the Range Space}
\label{sec:relationvcpredicate}

In order to actually deploy the VC dimension bounds previously discussed it is necessary to bound the VC dimension of the class of queries being considered. While challenging in general, we develop here a simple and effective bound on the VC dimension of the class of queries being considered based on the complexity of the constraints defining the predicates.

As discussed in Section~\ref{sec:model}, we assume that the values of the features can be mapped to real numbers. Hence constraint of values of a certain feature formalized using the operators $\geq, \leq, =$ and $\neq$ correspond to selecting intervals (either open or close) of the possible values of a feature. 
For each feature, the various clauses are \emph{connected} by means of ``\texttt{or}'' operators. We characterize the complexity of such connection by the minimum number of \emph{non-redundant} open and close intervals of the value. In particular we say that a connection of intervals is  \emph{non-redundant} is there is no connection of fewer intervals that selects the same values. 

The VC dimension of a class of queries can then be characterized according to the number of non-redundant constraints applied to the various features. 

\begin{lemma}\label{lem:VCbound}
Let $\mathcal{Q}$ denote the class of query functions such that each query is a conjunction of connections of clauses on the value of distinct features. The VC dimension of $Q$ is:
\begin{equation}
VC\left(Q\right) = \sum_{i=1}^m 2\alpha_i + \beta_i,
\end{equation}
where $\alpha_i$ (resp., $\beta_i$) denotes the maximum number of non-redundant closed (resp., open) intervals of values  corresponding to the connection of constraints regarding the value of the $i$-th feature, for $1\leq i\leq m$.
\end{lemma}
The proof for Lemma~\ref{lem:VCbound}, presented in detail in Appendix~\ref{app:vcproof}, proceeds by induction of the number of features which can be used in constructing the queries, and builds on known results on the VC dimension of a class of functions constituted by the union of a finite number of closed intervals on $\mathbb{R}$. 

Algorithm~\ref{alg:reduction} (in Appendix~\ref{app:reduction}) outlines a procedure which reduces an input query class $Q$ to an equivalent non-redundant version and computed a bound on its VC dimension.

\subsection{Trade-off between query complexity and minimum allowable selectivity}\label{sec:VCselectivity}

Consider exploring the space of possible recommendations by constructing the filter condition one clause at a time. Note that as the filter condition grows in its complexity (i.e., multiple non-trivial clauses are added) the number of records selected by the predicate, and hence, it selectivity, will decrease. 
It appears therefore reasonable to start evaluating \emph{simpler} predicate filters and then proceed \emph{depth-first} by adding more and more clauses. While reasonable, such procedure will possibly lead to explore large number of queries. However, most of the filters obtained by composing a high number of filters will likely lead to visualizations supported by very little sample points, and, hence, intrinsically \emph{unreliable}. 

Our VC dimension approach allows the system to recognize this fact and to use it in order to limit the search space. As discussed in Section~\ref{sec:VCselectivity}, the lower  the selectivity $\gamma$ of the filter condition $F$ of a given visualization, the higher the uncertainty $\bar{\epsilon}$,
\begin{equation*}
	\bar{\epsilon} \geq  \sqrt{\frac{d+\log_2 \delta^{-1}}{2n}\gamma^{-1}}
\end{equation*}

 In order for the difference between a candidate and the starting visualization to be deemed \emph{statistically relevant} their Chebyshev distance has to be higher than $\bar{\epsilon}$. 
The Chebyshev norm distance, and, hence, the maximum \emph{interest} of a candidate visualization is one. This clearly implies that all visualizations whose selectivity $\gamma$ is such that 
\begin{equation}\label{eq:lwbsele}
\gamma \leq \frac{d+\log_2 \delta^{-1}}{2n}    
\end{equation}
are not going to be interesting according to our procedure, and, hence, when exploring the space of possible recommendations, we can stop refining the queries once the selectivity of the candidate visualization drops below the threshold given by~\eqref{eq:lwbsele}. This allows to \emph{prune} the search space by eliminating from the exploration queries which are ``\emph{not worth to be considered}'' as possible recommendations.
\\
\\
While this may appear as a weakness of our approach, it is instead consistent with the basic principle of distinguishing statistically relevant phenomena from effects of the noise introduced by the sampling process. While visualizations which involve just a very limited number of sample points may appear to represent a very interesting distribution of a subset of the data, they are more likely to represent random fluctuation in the selection of the input sample than a true phenomenon in the global domain. 

By taking into consideration the selectivity of candidate visualizations, our method \emph{automatically adjusts} the threshold of interest for candidate visualization. 

%% file: tradeoff.tex
\section{Discussion}\label{sec:discussion}
 In this section, we discuss and motivate some guidelines to help the analyst determine which of the discussed tools are better suited for her actual setting.
\vspace{-2mm}
\subsection{Number of hypotheses being tested}
Consider a scenario for which the system is limited to the analysis of  a small number of possible visualizations. In this case, it is possible to evaluate which of these candidate visualizations are actually interesting with respect to the starting visualization $V_1$ by applying the $\chi^2$ testing with opportune correction for the number of hypotheses being tested, which in this case would correspond to the number of candidate visualizations. 

Further, if in the same exploration session the analyst wants also to evaluate which of the visualizations are interesting with respect to a different starting visualization $V_2$, this would require to treat all the additional candidate recommendations as additional hypotheses. 
In order to obtain FWER it will be necessary to correct for the number of hypotheses being tested, thus resulting in a loss of statistical power. Even though some FWER corrections such as the Holm or the H\"ochberg procedure allows to reduce the effect of the multiple hypotheses correction, compared to the the simple Bonferroni procedure, there is still a considerable decrease in the statistical power, in particular for settings for which only a low fraction of the hypotheses being tested have low $p$-values. Therefore the $\chi^2$ testing approach appears to be more suitable for settings with a limited number of candidate recommendations being evaluated.

\subsection{Bounding the complexity of the query class}
The properties of our method can also be used to determine a bound the VC-dimension of the query range space being considered when looking to ensure that any candidate recommendation who differs from the reference visualization by at least $\theta$ is actually marked as a safe recommendation.

Let $\vis{1}$ (resp., $\vis{2}$) denote the reference (resp., a candidate) visualization, and let $\gamma_1$ (resp., $\gamma_2$) is selectivity.

Given the size of the available dataset $|\mdataset{}|$, and the desired FWER control level $\delta\in(0,1)$, the maximum VC dimension which guarantees to meet these requirements can be obtained from~\eqref{eq:lwbsele} as:
\begin{equation*}
    d\leq \theta^2 \min\{\gamma_1,\gamma_2\} n-\log_2(\delta^{-1}).
\end{equation*}

This bound can be used as a guideline to limit the structure of the queries being considered. That is, it offers indications on the number of different features which can be considered when building the queries, and indications on their \emph{complexity}, intended as the number of clauses being used in their construction (as discussed in Section~\ref{sec:relationvcpredicate}).

%% file: new_strategies.tex
\subsection{Preprocessing heuristics}
In this section we outline some preprocessing heuristics which allow to improve the effectiveness of our control procedures. 

\textbf{I. Removal of constant features}: Features which assume the same value in all the records of the sample can be safely ignored.

\textbf{II. Removal of identifier-type features}: Features which assign a different \emph{unique} value to each of the record can be removed. This is generally the case for \emph{identifier} features (e.g.,``Street address'' for real estate dataset). This appears justified as, due to the uniqueness of their values, they are not useful in constructing predicate conditions, nor they should be used as the ``\emph{group-by}'' attribute (i.e., the attribute in the x-axis). 
All these heuristics share the fact that they allow to ignore some of the features (or columns) of the records. This will in turn impact both the search space and will allow to reduce the VC dimension of the query class being considered.

%% file: experiments_new.tex
\section{Experiments}\label{sec:experiment}
In this section, we want to show how our framework can be applied towards both real data (i.e. the collected survey data) and synthetic data. We start by demonstrating different, problematic scenarios with our real dataset.
\subsection{Anecdotal examples}
\label{sec:anecdotal_examples}
Our first example shows that a system without statistical control may trick the user to believe in insights that are actually not valid and merely random. For this, assume the user wants to explore whether there is a subpopulation that believes differently from the overall population when it comes to whether obesity is a disease or not.

\begin{figure}[ht!]
\centering
    \begin{subfigure}[b]{0.3\textwidth}
        \centering
        \includegraphics[width=\textwidth]{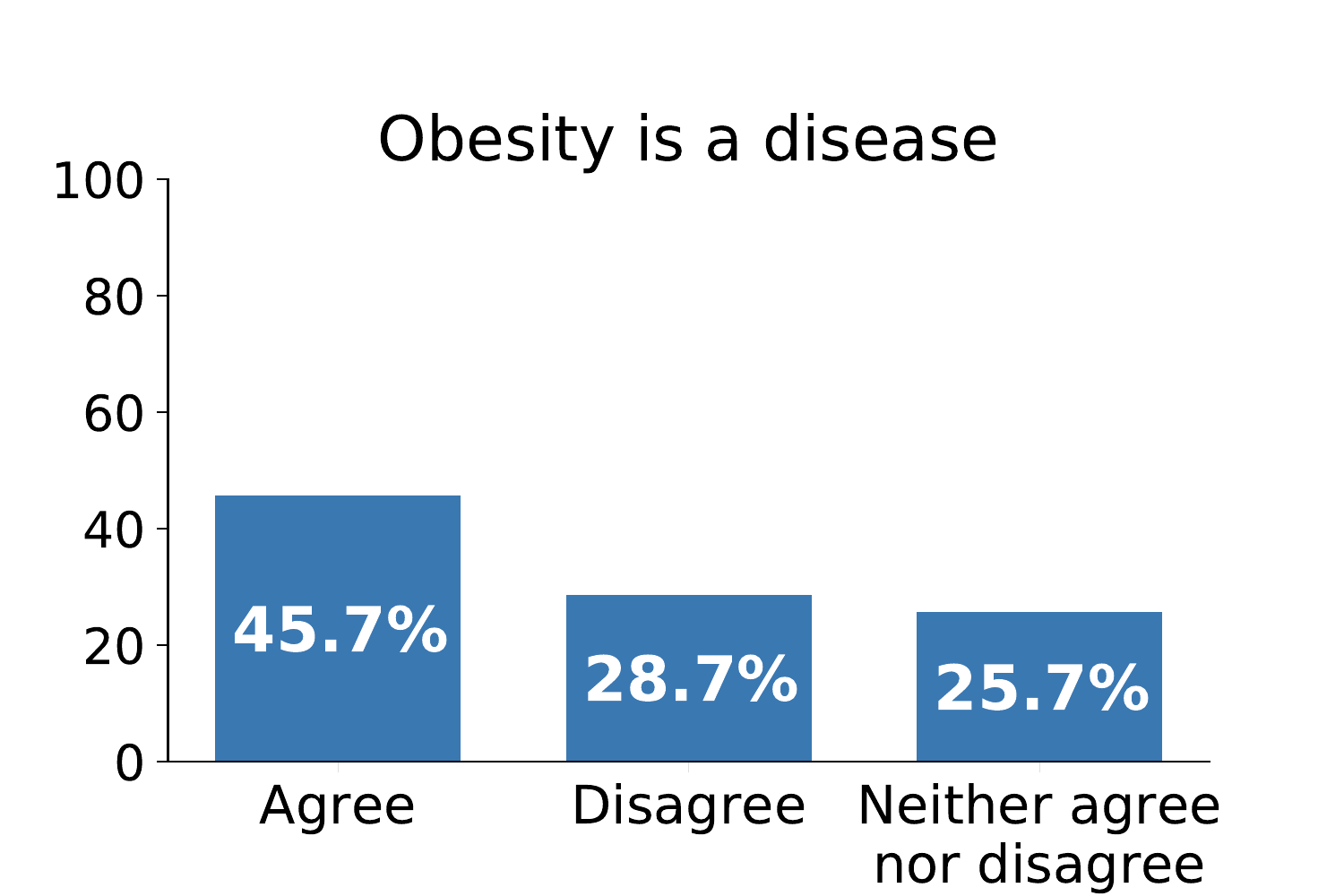}
        \caption{Reference View}
    \end{subfigure}
    \begin{subfigure}[b]{0.3\textwidth}
        \centering
        \includegraphics[width=\textwidth]{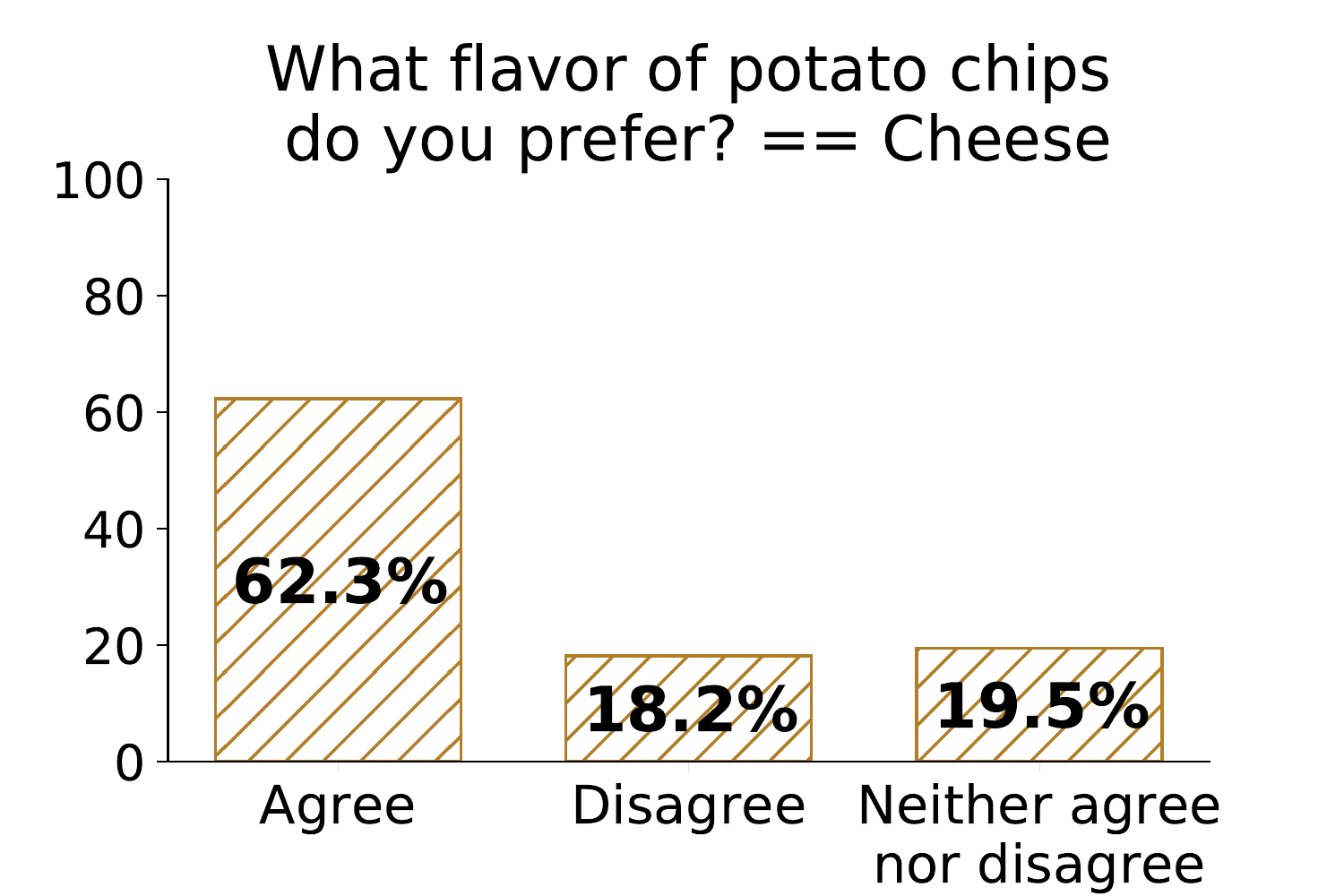}
        \caption{SeeDB View 1}
    \end{subfigure}
    \\
    \begin{subfigure}[b]{0.3\textwidth}
        \centering
        \includegraphics[width=\textwidth]{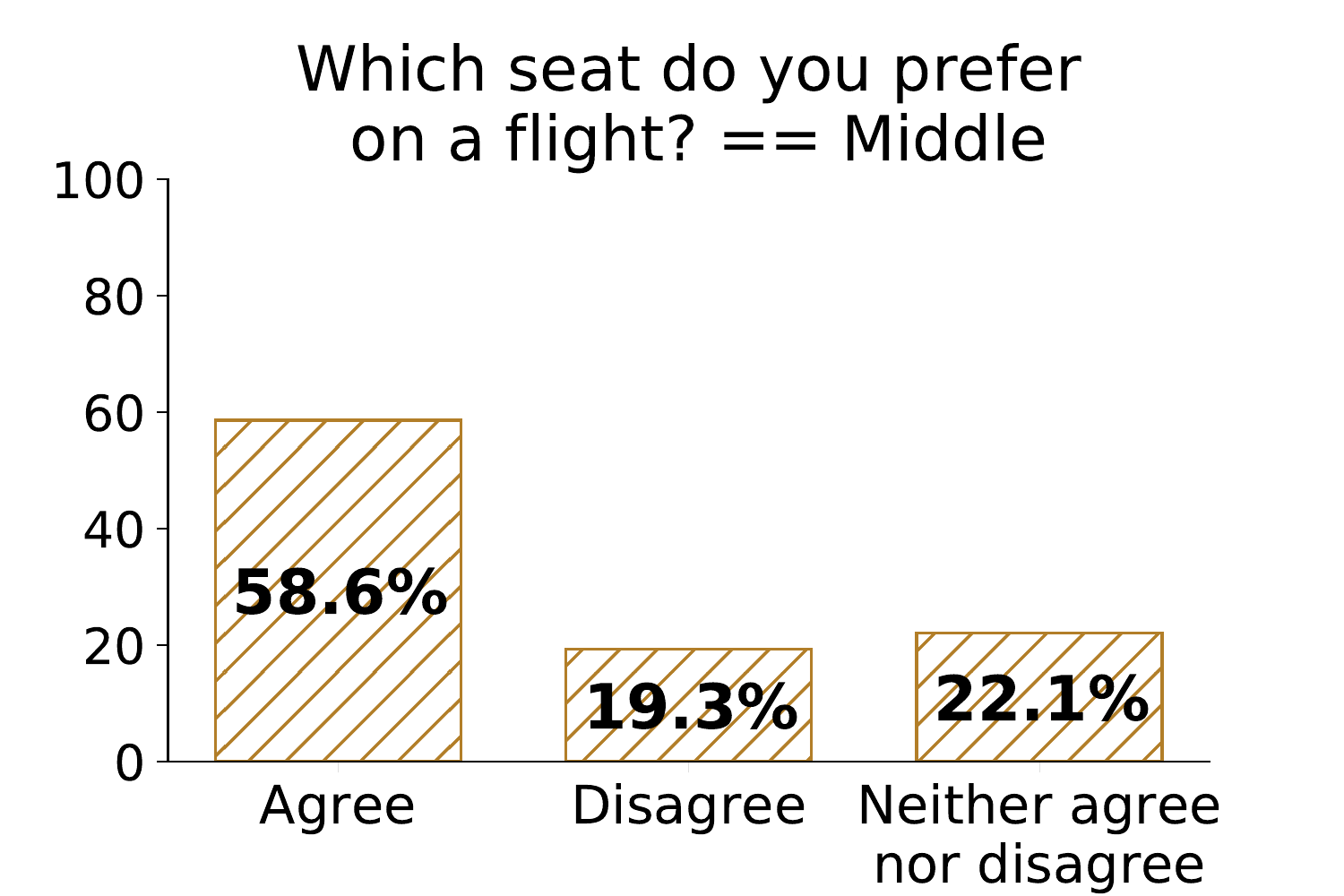}
        \caption{SeeDB View 2}
    \end{subfigure}
    \begin{subfigure}[b]{0.3\textwidth}
        \centering
        \includegraphics[width=\textwidth]{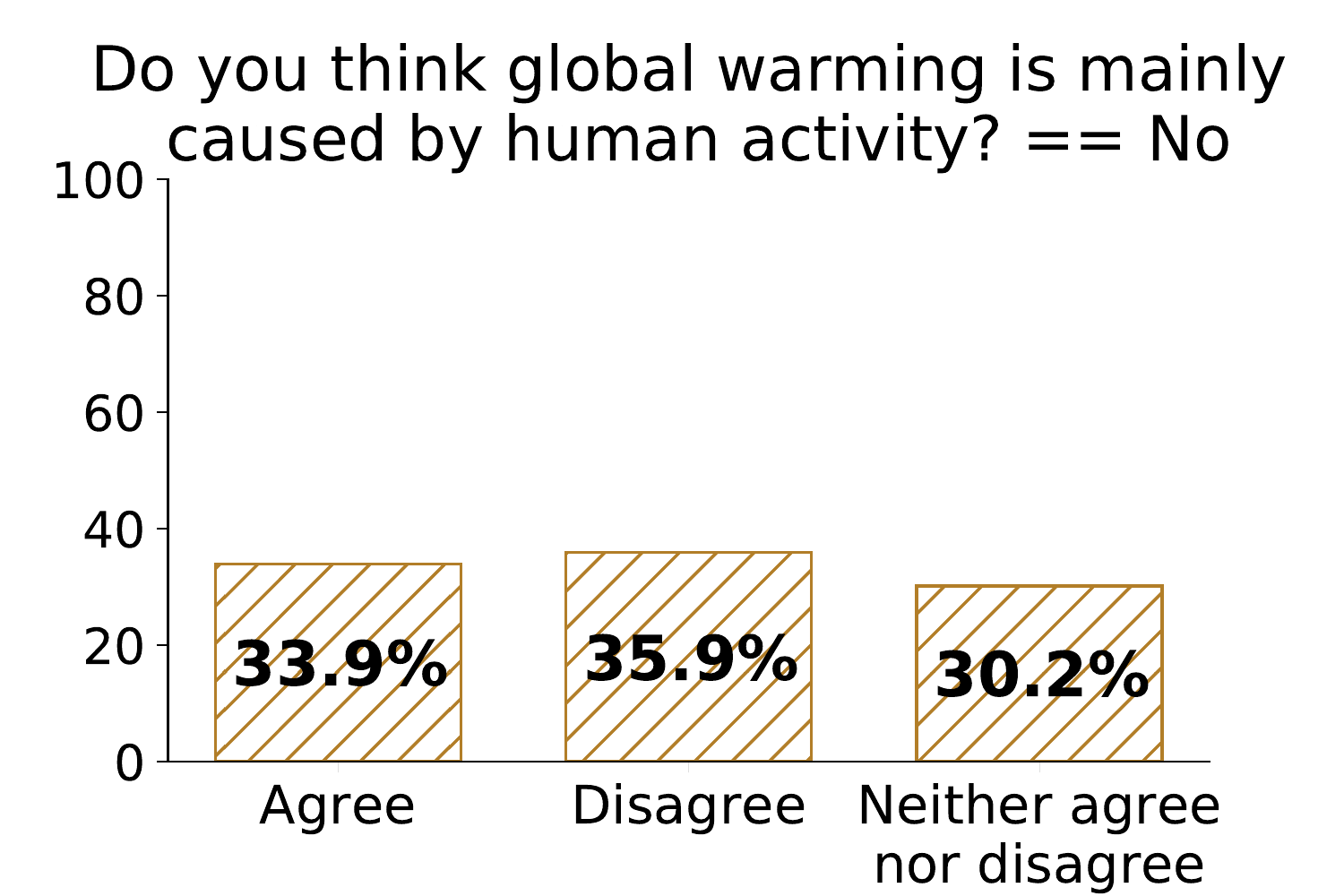}
        \caption{SeeDB View 3}
    \end{subfigure}
    \caption{VizRec would not mark any of these visualizations as statistically significant as the difference computed with respect to the reference is not larger than the uncertainty of estimating the bars correctly.}
    \label{fig:anecdotal_no_rec}
\end{figure}
As depicted in \autoref{fig:anecdotal_no_rec} a user may falsely believe that people who prefer potato chips with Cheese flavour are more likely to believe that obesity is a disease. Though people that consume potato chips may lean more towards a view that obesity is a disease, the insight that in particular people who prefer the Cheese flavour are the most interesting subpopulation is questionable. Since for all the other flavours in our study\footnote{BBQ, Sea Salt \& Vinegar, Sour Cream and Onion, Jalapeno, Cheddar and Sour Cream, Original/Plain, I don't eat chips} no visualization is within the top results, picking particularly the Cheese flavour as a belief changer looks like a potential false discovery. The next recommended result seems even more random: An automatic recommender system would imply to the user that persons who prefer the middle seat are more likely to believe that obesity is a disease. This seems hard to understand and more like some random result that the system produced.
\\
\\
In our second example, we want to show that VizRec may identify correctly the top SeeDB recommendation(s) as being statistically valid.
\begin{figure}[H]
\centering
    \begin{subfigure}[b]{0.3\textwidth}
        \centering
        \includegraphics[width=\textwidth]{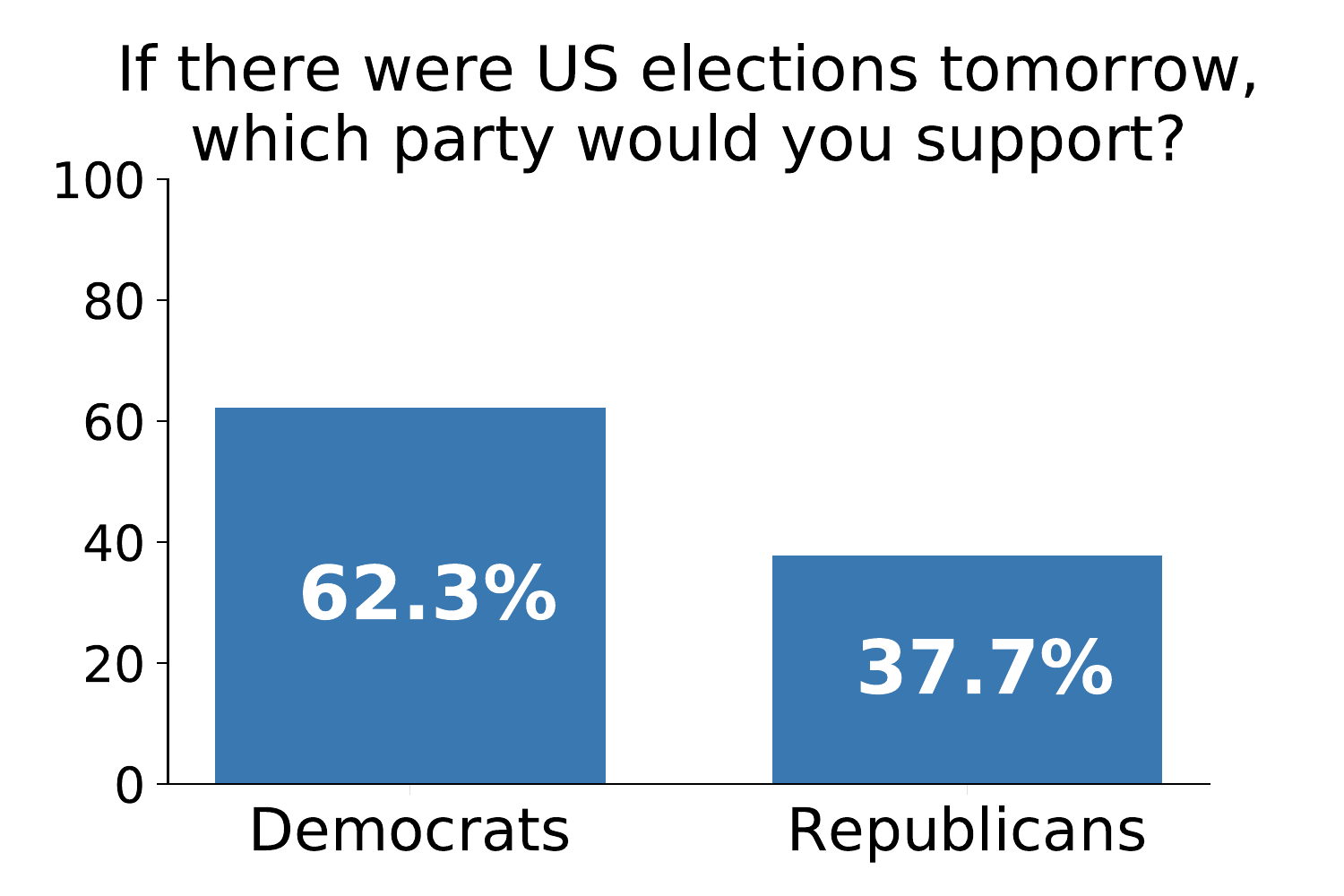}
        \caption{Reference View}
    \end{subfigure}
    \begin{subfigure}[b]{0.3\textwidth}
        \centering
        \includegraphics[width=\textwidth]{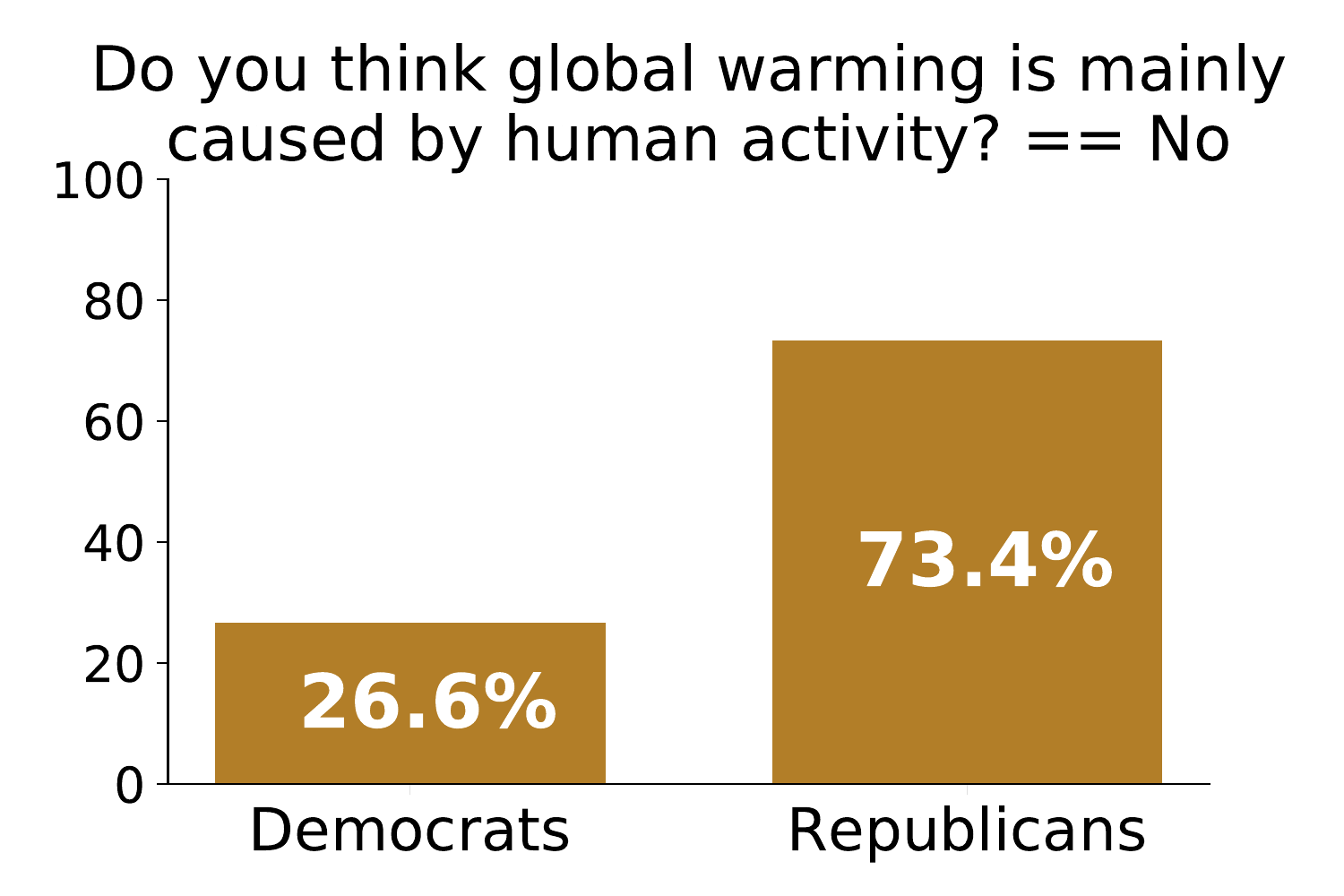}
        \caption{SeeDB View 1}
    \end{subfigure}
    \\
    \begin{subfigure}[b]{0.3\textwidth}
        \centering
        \includegraphics[width=\textwidth]{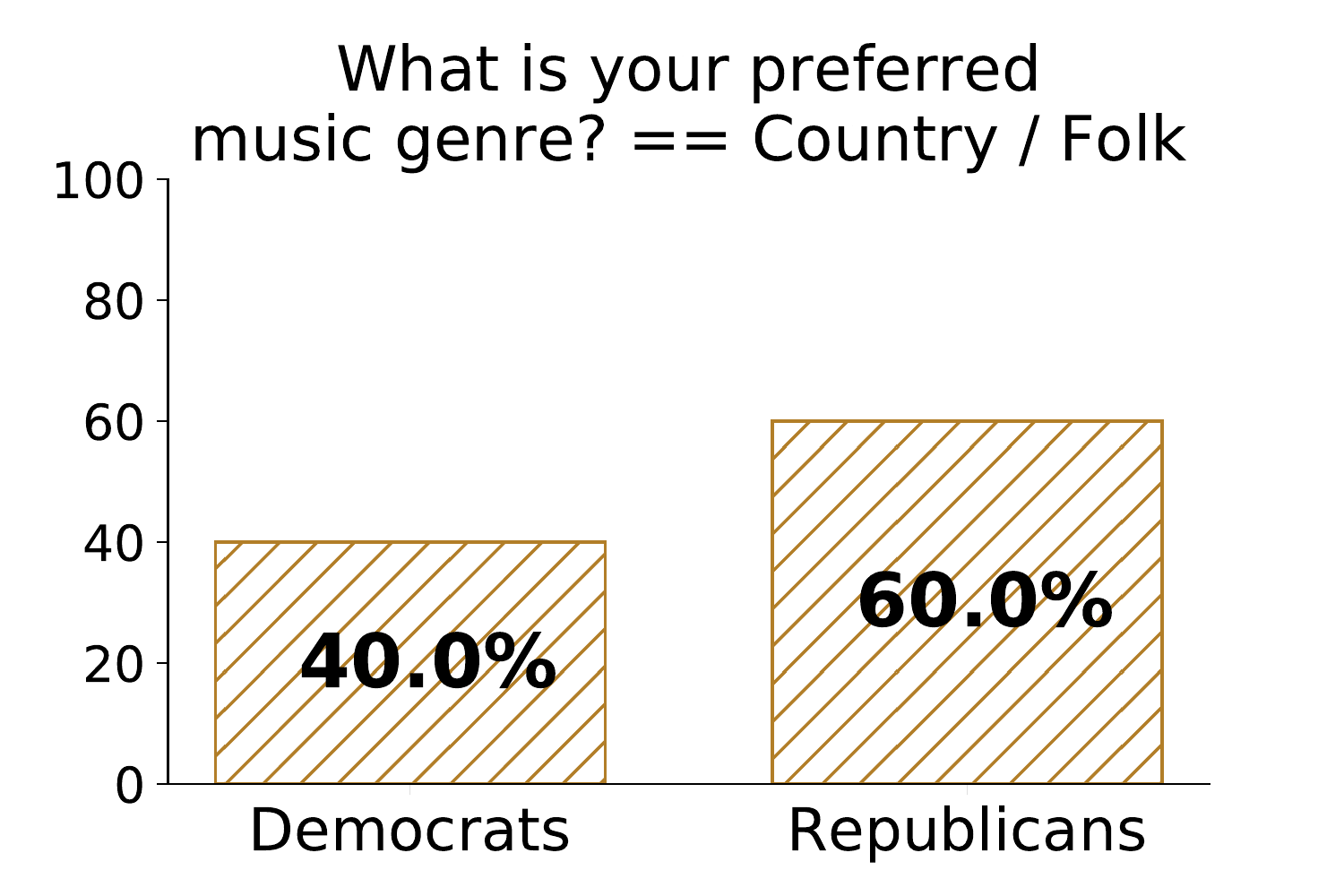}
        \caption{SeeDB View 2}
    \end{subfigure}
    \begin{subfigure}[b]{0.3\textwidth}
        \centering
        \includegraphics[width=\textwidth]{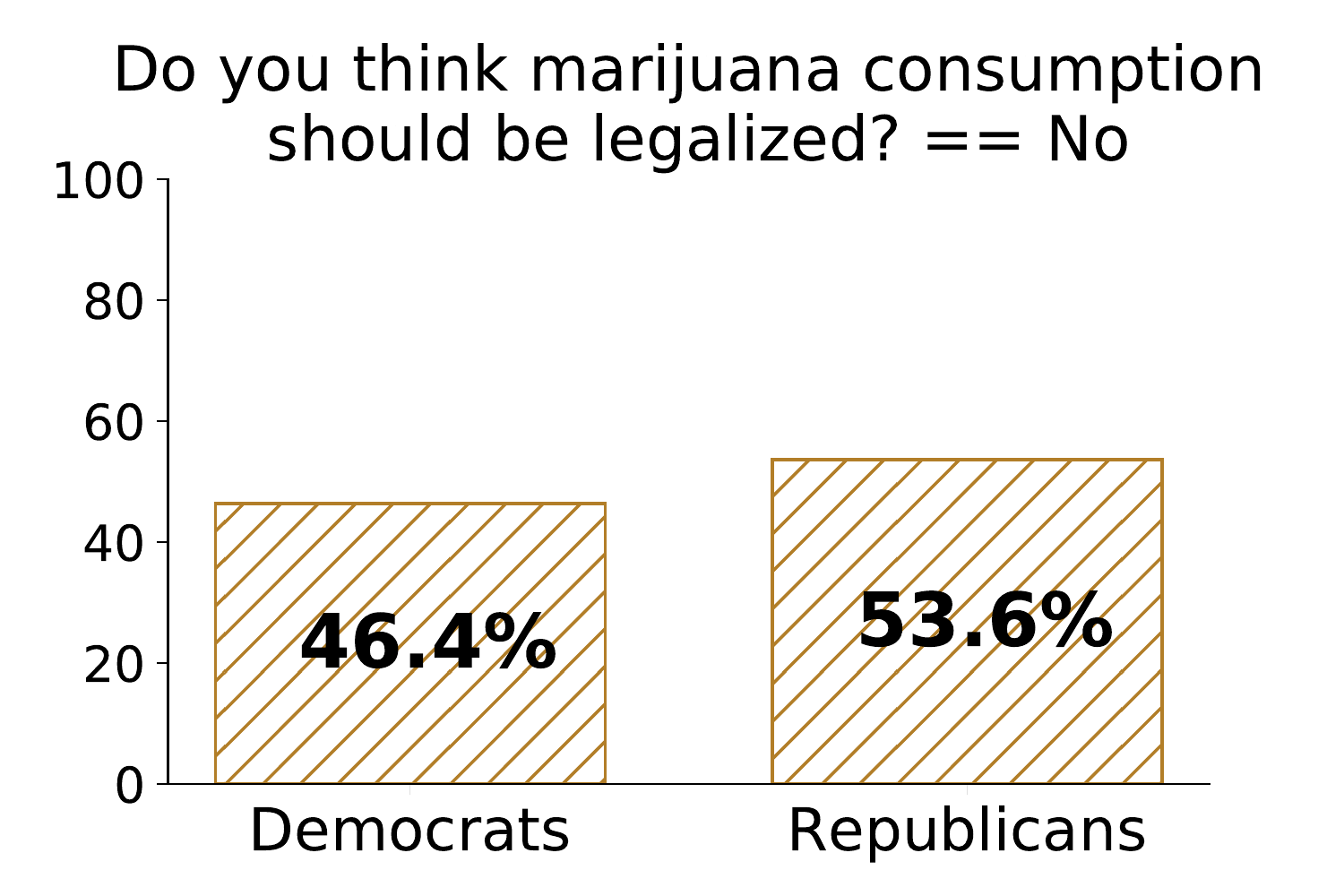}
        \caption{SeeDB View 3}
    \end{subfigure}
    \caption{VizRec would also recommend the top visualization, but declares the other visualizations not being statistically significant enough.}
    \label{fig:anecdotal_top_same}
\end{figure}

Here, the user was interested in finding out whether there is a subpopulation that has a different voting behavior for the two main parties in the United States(cf. \autoref{fig:anecdotal_top_same}). The recommended top visualization coincides with the top SeeDB result and seems sound. However, VizRec prevents the user to attribute a preference towards the Republican party for persons who prefer to listen to Country and Folk music.
\\
\\
Finally, it can also be the case that SeeDB recommends something as the top result which the VizRec framework would rule out as being not significant at all. As depicted in \autoref{fig:anecdotal_top_different} again a questionable relation between people who prefer Cheese flavoured potato chips and those who belief in Astrology would get recommended.
\begin{figure}[H]
\centering
    \begin{subfigure}[b]{0.3\textwidth}
        \centering
        \includegraphics[width=\textwidth]{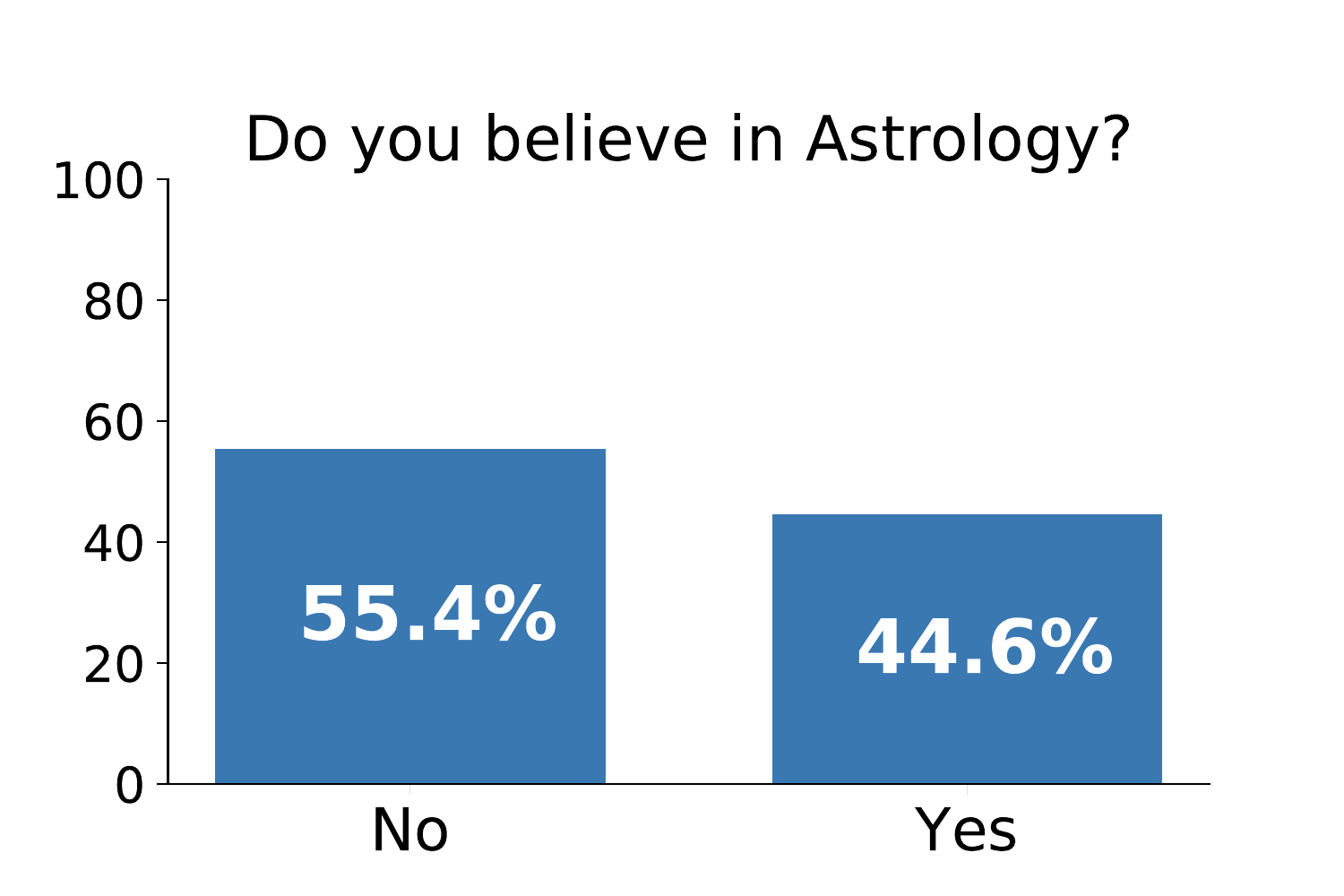}
        \caption{Reference View}
    \end{subfigure}
    \begin{subfigure}[b]{0.3\textwidth}
        \centering
        \includegraphics[width=\textwidth]{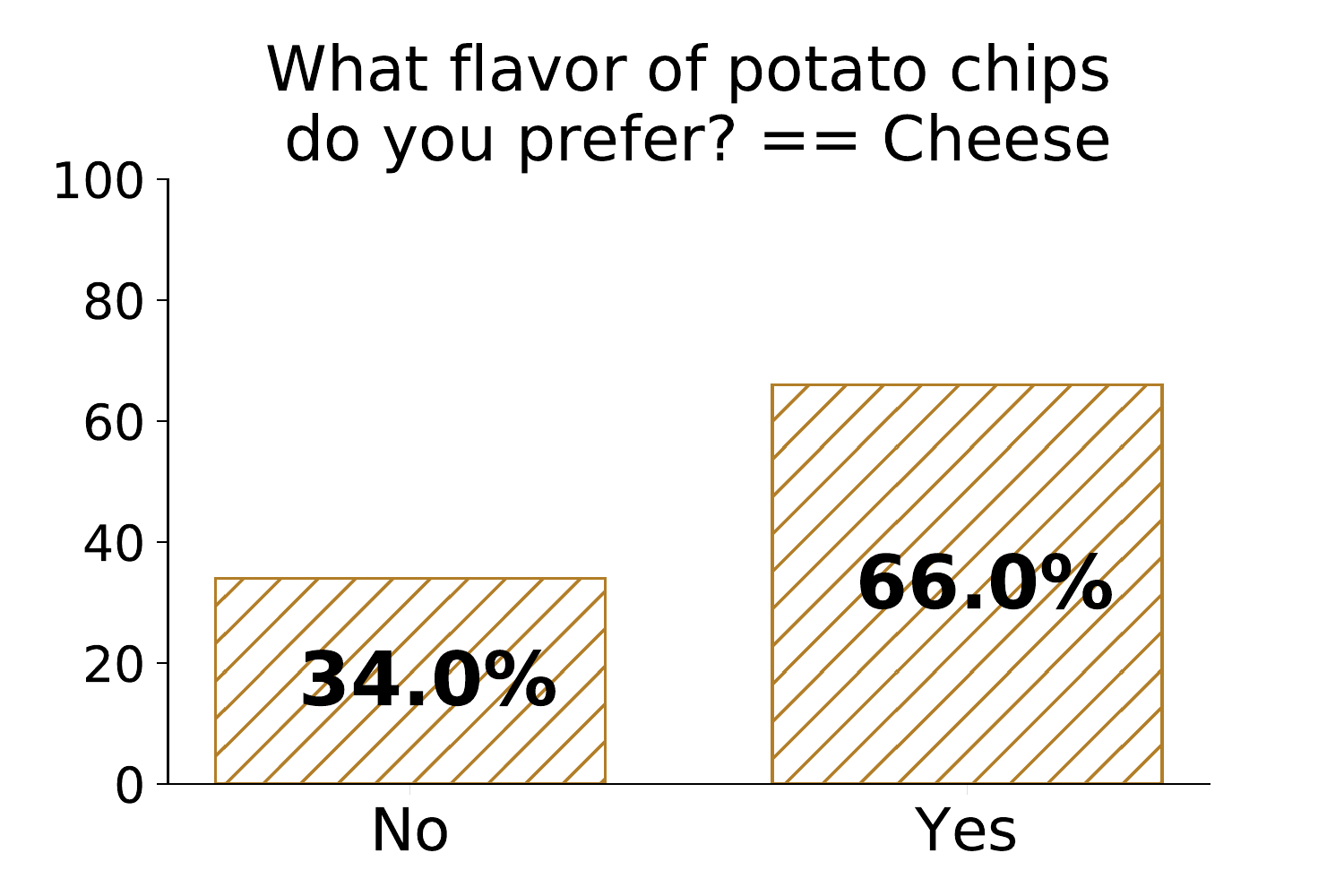}
        \caption{SeeDB View 1}
    \end{subfigure}
    \\
    \begin{subfigure}[b]{0.3\textwidth}
        \centering
        \includegraphics[width=\textwidth]{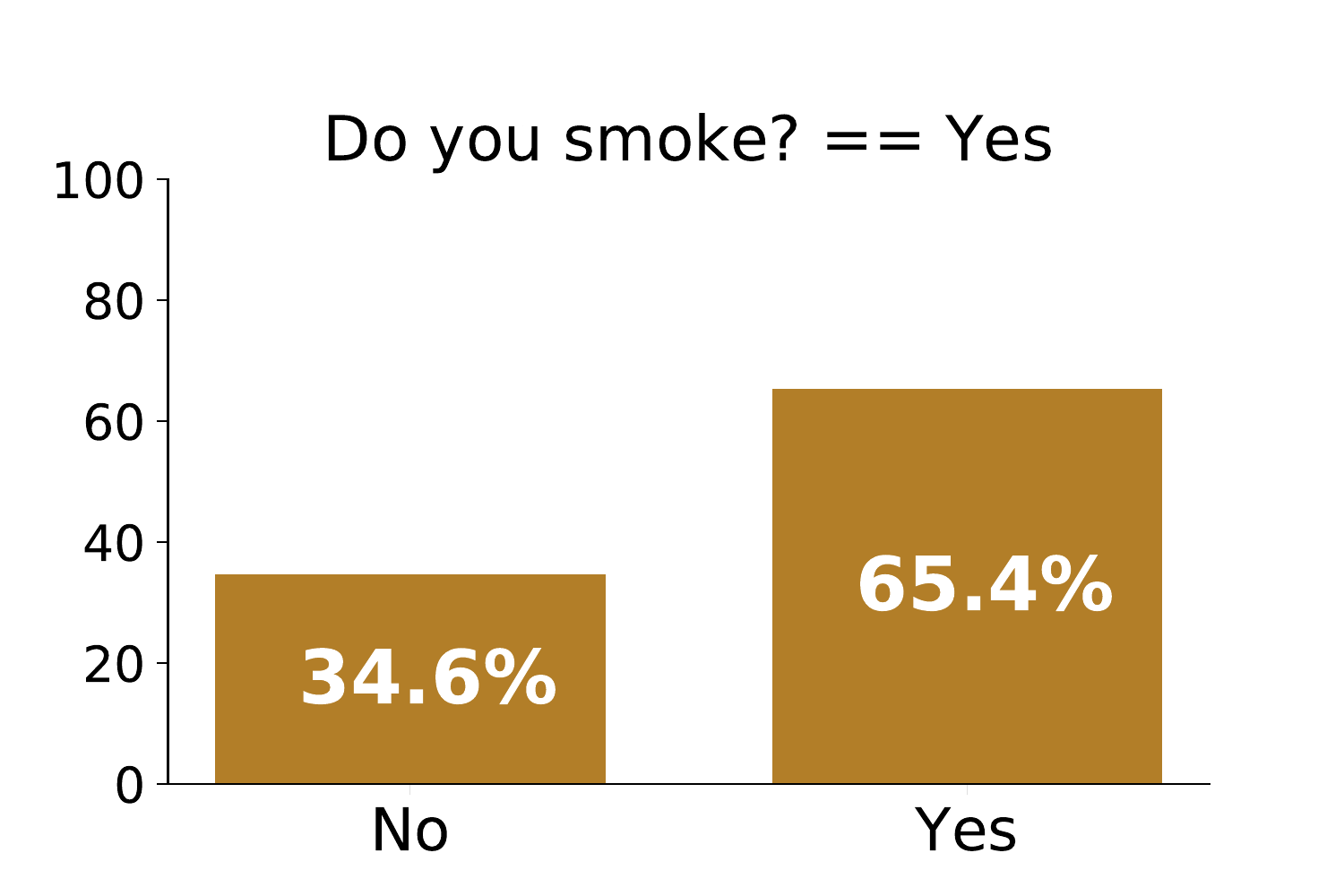}
        \caption{SeeDB View 2}
    \end{subfigure}
    \begin{subfigure}[b]{0.3\textwidth}
        \centering
        \includegraphics[width=\textwidth]{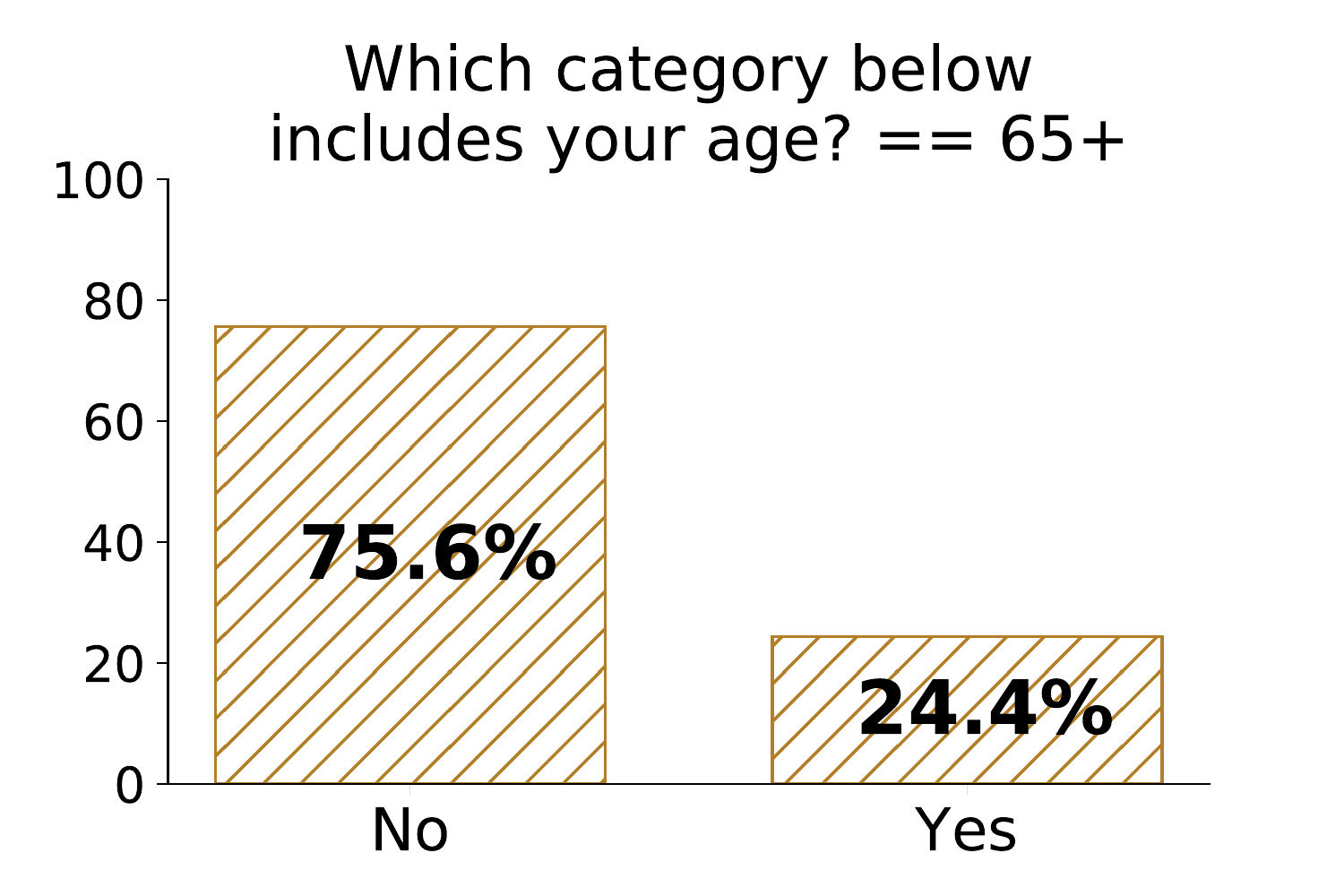}
        \caption{SeeDB View 3}
    \end{subfigure}
    \caption{VizRec would not recommend the top visualization, but the second ranked one.}
    \label{fig:anecdotal_top_different}
\end{figure}
However, this is actually a false discovery and not backed statistically. On the contrary, VizRec deemed a relation between smokers and a belief in Astrology statistically sound. In this example another interesting problem is demonstrated: Though the interestingness scores of the top results are pretty close to each other (e.g. here $d_\infty \approx 0.21$), the score $d_\infty$ itself is not sufficient to determine statistical relevance. Particularly, employing a simple cutoff may lead to false discoveries. Besides the complexity of the exploration space, the number of samples used to estimate both the reference query and the candidate query need to be accounted for too.
\\
\\
These anecdotal examples demonstrate that without any statistical control the user is likely to run into making false discoveries. Our VC approach does not only account for the number of samples but also for the complexity of the data exploration and is thus a well-suited tool for avoiding false discoveries in a visual recommendation system.
\subsection{Random data leads to no discoveries}
A meaningful baseline for any safe visual recommendation system is to make sure that random data does not lead to any recommendations. To demonstrate that the VC approach will not recommend any false positives, we generated a synthetic dataset with uniformly distributed data. $100,000$ samples were generated in total with the first column being selected as aggregate and the other 3 columns as features. The aggregate is uniformly distributed over $\lbrace 1, 2, 3, 4\rbrace$ and each of the $3$ features are uniformly distributed over $\lbrace 1, ..., 9 \rbrace$.

With simple predicates (i.e. a queries formed from $\leq$ clauses solely) there are $1331$ visualizations to be explored (a dummy value of $+\infty$ was used in the queries to make a feature active or not. E.g. consider a query of the form $\left( X_1 \leq 8 \right) \land \left( X_2 \leq +\infty \right) \land \left( X_3 \leq 3 \right)$. In this query, feature $X_2$ has no effect on the rows returned since $\left( X_1 \leq 8 \right) \land \left( X_2 \leq +\infty \right) \land \left( X_3 \leq 3 \right) \equiv \left( X_1 \leq 8 \right) \land \left( X_3 \leq 3 \right)$. Note that using $+\infty$-values in the clauses does not change the VC dimension.). As a reference, a uniform distribution over $\lbrace 1, 2, 3, 4\rbrace$ was chosen. This means, that the expected support of any visualization is at least $10^5/9^3$ samples which is a fair amount to estimate $4$ bars.
\begin{figure}[ht!]
\centering
    \includegraphics[width=.8\columnwidth]{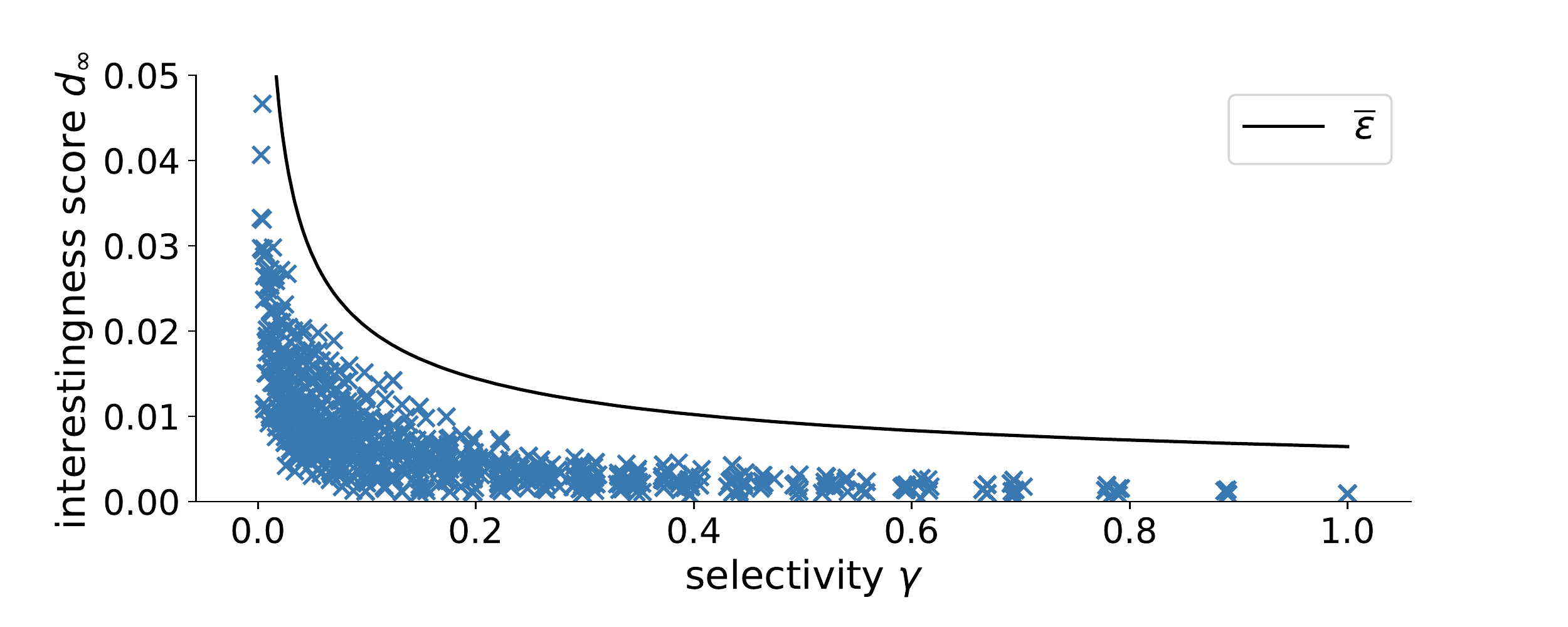}
\caption{Blue dots represent interest scores all evaluated visualizations. The  $\bar{\epsilon}$ curve denotes the threshold for recommendation using VC dimension = 4 to achieve control at level $\delta = 0.05$. The lower the VC dimension the more the $\bar{\epsilon}$ curve takes the form of an ``L''. Since there are no visualizations with scores higher than the $\overline{\epsilon}$-curve, no visualizations get recommended from the generated random data.
}
\label{fig:exp1_vc}
\end{figure}
When not accounting for the multiple comparison problem $p$-values below the threshold of $\alpha=0.05$ occur inevitably. A system without FWER guarantees would classify them thus as false positives. Using Bonferroni (or other comparable corrections) remedies this while, however, incurring a noticeable loss in statistical power.

In comparison, the lowest $\epsilon$ the VC approach guarantees is $\epsilon_{\min} = 0.0059$. As discussed in \ref{ConcreteVC} the required threshold $\overline{\epsilon}$ to be met by the Chebychev norm induced distance measure $d_\infty$ depends on the selectivity $\gamma$ of the query. The necessity of this can be observed in \autoref{fig:exp1_vc} too. With the interestingness scores(distances) being lower than the curve defined by $\overline{\epsilon}$ for all queries in \autoref{fig:exp1_vc} the VC approach does not recommend any false positives in this experiment. Using different distributions instead of the uniform one showed comparable results.

\subsection{Statistical Testing vs. VC approach}
\label{sec:exp_chisq}
We now show that while statistical testing in the form of a \chisqtest is in general not a wrong ingredient for building a VRS, in some situations it is unable to spot meaningful visual differences which would however be opportunely recognized by our \system{} approach.

Assume we had a query that yielded $m=1,200$ out of $n=10,000$ samples and a perfect estimator for the true distribution function of the reference and the query distribution. These distributions shall be distributed as in \autoref{fig:exp2_chi2vsvc}. Thus, the \chisqtest would yield a p-value of of $2.54 \cdot 10^{-5}$ implying that they are different when no more than $1967$ visualizations under Bonferroni's correction are tested.
\begin{figure}
\centering
     \includegraphics[width=.7\columnwidth]{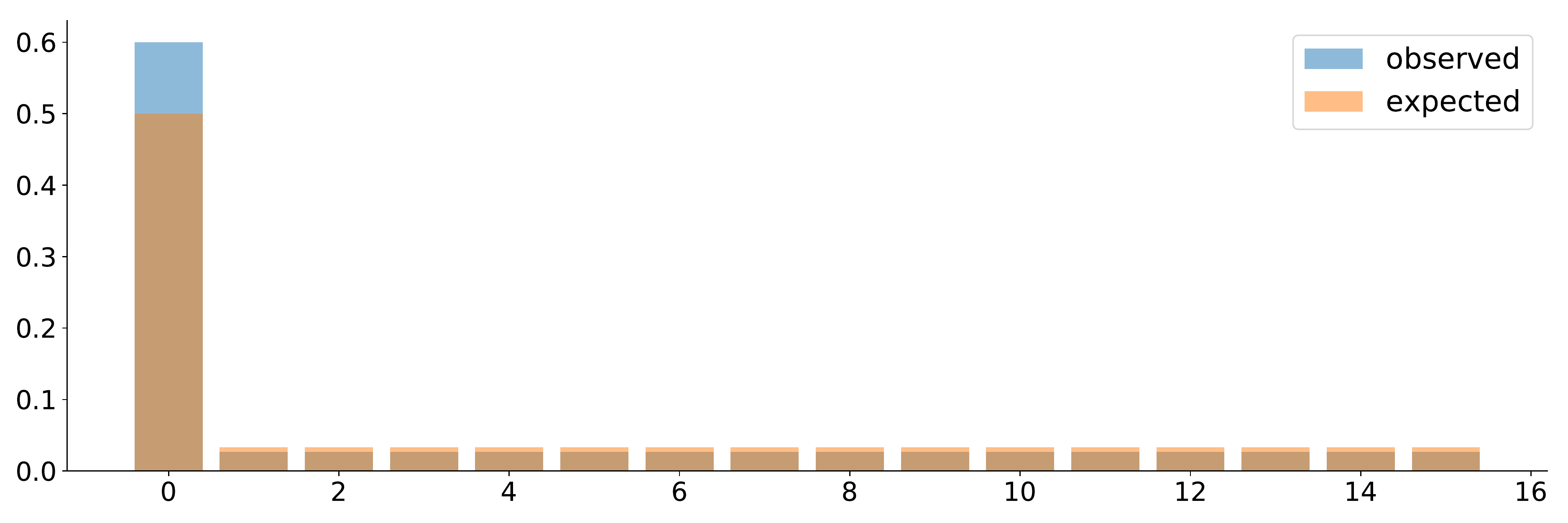}
\caption{Comparing two close distributions that however should not be recommended since the visual difference criterion according to the VC dimension approach is not met.}
\label{fig:exp2_chi2vsvc}
\vspace{-5mm}
\end{figure}
However, at a VC dimension of $10$ ($\delta=0.05$) the required $\overline{\epsilon}$ must be at least $0.22$ which is nearly twice as high as the $0.1$ difference at the first bar as shown in \autoref{fig:exp2_chi2vsvc}. Thus, the VC approach would not select this visualization as being significantly different enough given the modest sample size. Using the $\chi^2$ test would recommend this visualization since it only spots that there is a difference but not whether the difference is significant enough. 

In practice, scenarios like the one presented here occur especially due to outliers in the data. I.e. it could happen that for one feature value there are only 1-5 samples that would lead without any correction to a positive recommendation. Though a heuristic may ignore visualizations with less than $5$ samples, this would come at the cost of ignoring rare phenomena and by using some magic number to define the threshold. I.e. when ignoring visualizations with less than $b$ samples per bin, easily an example with $b+1$ samples for some bin could be found where $\chi^2$ would pick up a non-visually significant visualization.

This underscores that a VRS using the \chisqtest{} would correctly identify two visualizations being different but can not guarantee a meaningful difference in terms of a distance which is crucial to build usable systems without luring the user into a false sense of security. One may argue that filtering out visualizations after having performed statistical testing would remedy this(which may work in practice when the interestingness score is high enough), but then there was no guarantee that the distances observed is statistically guaranteed.
Whereas the question is dependent on the scenario (i.e. which query and which guarantees a system needs to fulfill) we want to point out, that there is a simple way to use the \chisqtest to control significant distance too.
\\
\\
Furthermore we want to underscore the point that the Chi-squared test is indeed a very powerful test but that the correct estimation of the distribution dominates the selectivity. 
I.e., when we guarantee that the estimates for the probability mass function are close enough to the true values, a testing procedure like \chisqtest will even under a million possible hypothesis only need a small number of samples to spot a difference between two distributions. We thereby define the required number of point estimates to be in the range of $2 \leq K \leq 100$ bars as meaningful.
\begin{figure}[ht!]
\centering
     \includegraphics[width=.7\columnwidth]{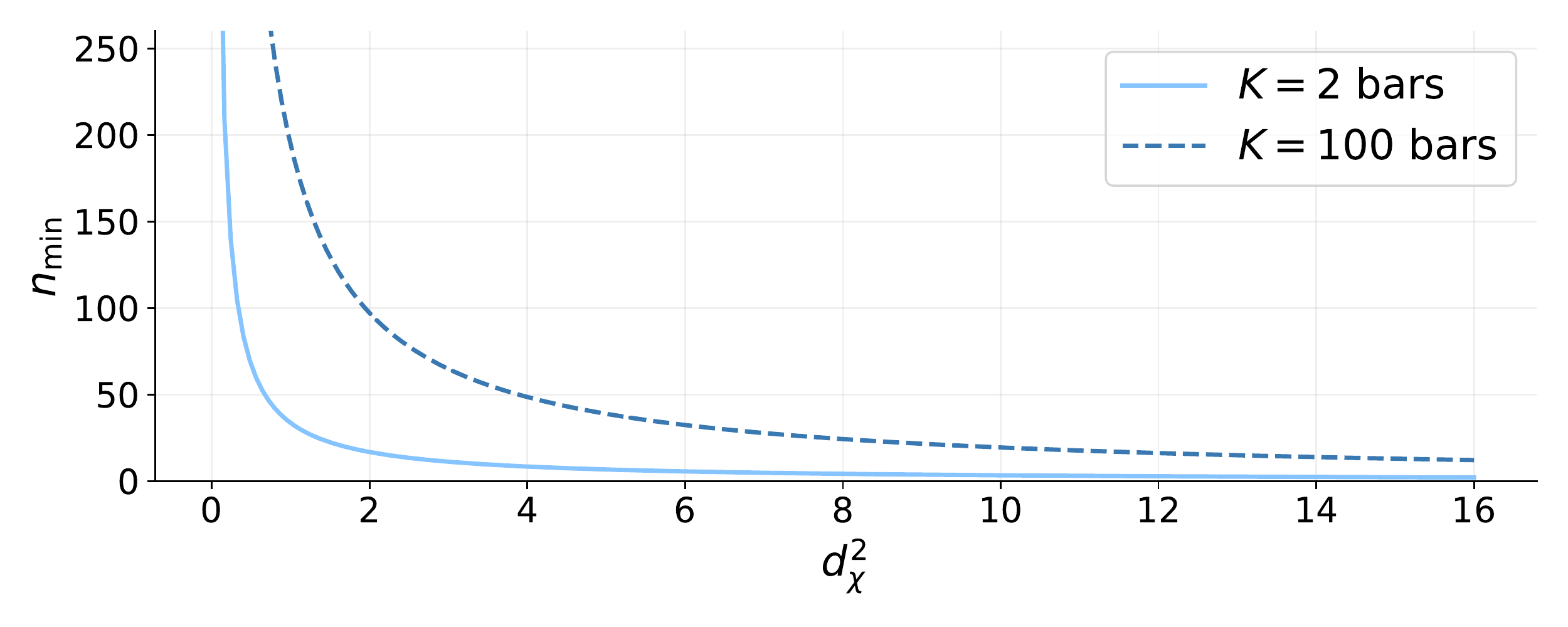}
\caption{Chi-square distance $d_{\chi^2}$ and minimum number of samples $n_{\mathrm{min}}$ required for the \chisqtest to reject the null hypothesis assuming Bonferroni correction with $\alpha' = 0.05$ and $10^6$ queries.}
\label{fig:exp3_minchi2}
\vspace{-2mm}
\end{figure}
In \autoref{fig:exp3_minchi2} it is shown that even low values for the $\chi^2$-distance $d_\chi^2$ only require queries with hundreds of samples to be identified correctly.

%% file: related.tex
\section{Related Work}
\cite{matteovc} introduced the VC approach to provide $\epsilon$-approximations for the selectivity of queries. Whereas they also consider joins in addition to multi-attribute selection queries, by restricting to $\mathrm{AND}$ conjunctions over multiple attributes as used naturally in OLAP and visual recommender systems we were able to lower the required VC dimension.

When comparing continuous aggregates for visualizations using e.g. kernel-density estimation for distribution plots other statistical testing procedures like the Kolmogorov-Smirnov test can be used. Also a viable alternative to goodness-of-fit tests like Fisher's exact test, \chisqtest or the Kolmogorov-Smirnov test independence tests like Kendall's tau or Spearman's rank correlation test may be used. A combination of multiple tests also accounting for differences in shape as detailed in \cite{porter2008testing} can strengthen statistical guarantees on visualizations being different. Indeed, this is quite similar to hypothesis based feature extraction and selection approaches as in \cite{CHRIST201872}\cite{STEPPE199647}.

Recent work \cite{salimi2018hypdb} introduced the problem of group-by queries leading to wrong interpretations, specifically in the case when \sqlinline/AVG/ aggregates are used. To remedy this, the notion of a biased query is introduced. However, they do not account for the multiple comparison problem and also have no significant distance notion.

\cite{CIDRBrown} introduced various control techniques for interactive data exploration scenarios. Whereas it accounts for the multiple comparison problem, it does not solve the problem of pointing out a statistical different enough distance between two visualizations.

\cite{vartak2014seedb} provides an approach to effectively compute visualizations over an exponential search space by using reuse of previous results and approximate queries. Visualizations are recommended by treating group-by results as normalized probability distributions and using various distance measures between two probability distributions to yield a ranking in order to recommend top-$k$ interesting visualizations. The authors found that the actual choice of the distance did not really alter results, which does not come at a great surprise given their relations as pointed out in \cite{gibbs2002choosing}. 

As described in \cite{shneiderman1996eyes} zooming into particular interesting regions of the data is a key task performed by many users in the setting of data exploration. Our technique provides a simple and effective methodology which can be applied to a wide range of data. For example, an user may be given a geospatial dataset characterizing purchase power and want to have assistance in exploring interesting subregions. We believe our VC approach can be easily extended to allow for more complicated query types such as these.

%% file: conclusion.tex
\section{Conclusion}
In this work, we demonstrated why users should build visualization recommendation systems with mechanisms to ensure statistical guarantees in order to prevent users from making false discoveries or regarding noisy data as relevant. As a novel way which supplements classical statistical testing as in \cite{controlfd} we introduced a technique based on statistical learning comparable in its power to \cite{matteovc}. We demonstrated various trade-offs and problems to consider when using either technique and provided a simple heuristic on how to explore the vast search space more efficiently by pruning it through different preprocessing steps.

%% file: acknowledgments.tex
\section{Acknowledgments}
This research was funded in part by the DARPA Award 16-43-D3M-FP-040,  NSF Award IIS-1562657, NSF Award RI-1813444 and gifts from Google, Microsoft and Intel.

%% file: appendix.tex
\appendix
\section{Proof on VC dimension bound for classes of queries}\label{app:vcproof}
\begin{proof}[Proof of Lemma~\ref{lem:VCbound}]
The proof is by induction on $i$: in the base case we have $i=1$. In this case, the VC dimension of $\mathcal{Q}$ corresponds to the VC dimension of the union of $\alpha_1$ closed intervals and $\beta_1$ open intervals on the line. By a simple modification of the result of Lemma~\ref{lem:intervalunions}, we have that it has VC dimension at most $2\alpha_1 + \beta_1$.
Let us now inductively assume that the statement holds for $i>1$. In order to conclude the proof we shall verify that it holds for $i+1$ as well.

Assume towards contradiction that there exists a set $X$ of $\sum_{j=1}^{i+1}2\alpha_j+\beta_j$ points that can be shattered by $\mathcal{Q}$.
From the inductive hypothesis, we have that for any subset of $X$ with more than $\sum_{j=1}^{i}2\alpha_j+\beta_j$ cannot be shattered by the family of query functions which can express constraints only on the features $1,2,\ldots,i$. 
Without loss of generality let $X'$ denote one of the maximal subsets of $X$ which can be shattered using only the constraints on the features $1,2,\ldots,i$. The following fact is important for our argument

\textbf{Fact 1:} Recall that the queries in $\mathcal{Q}$ are constituted by logical conjunctions (i.e., ``$\texttt{and}$'') of  connections (i.e.,``$\texttt{or}$'' statement) of constraints on a feature. Hence, for any function in $\mathcal{Q}$ if any of the $i+1$ connections are such that they assume value ``$\texttt{false}$'', then the query will not select such point \emph{regardless} of the value of the remaining $i$ connections being conjuncted.

Consider any assignment $\pi$ of $\{0,1\}$ to the points in $X'$ and let $r_{\pi}$ the range which realizes such shattering.

If $r_\pi$ would assign to any point in $X\setminus X'$ value ``0'', then, according to the structure of the queries, no constraint on the $(i+1)$-th feature would allow to assign to it value ``1'', and, hence, it would not be possible to shatter $X$.

Note that for any assignment $\pi$ of $\{0,1\}$ to the points in $X'$, there may may not exists two ranges $r_1$ and $r_2$ such that based solely on constraints on the first $i$ features, on would assign ``0'' to a point in  $X\setminus X'$ and the other would assign ``1'' to the same point. If that would be the case, then i would be possible to shatter $\sum_{j=1}^i 2\alpha_j+\beta_j$ points using just constraints on the first $i$ features and this would violate the inductive hypothesis. 

Without loss of generality, in the following we can therefore assume that for any assignment $\pi$ of $\{0,1\}$ to the points in $X'$ the ranges that realize such assignment just based on the first $i$ features would assign ``1'' to all the points in $X\setminus X'$. This implies that the shattering of the points in $X\setminus X'$ relies \emph{solely} on the constraints on the values of the $i+1$-th feature. 

Consider now the points in $X\setminus X'$, according to our assumption $|X\setminus X'|= 2\alpha_{i+1}+\beta_{i+1}$. As discusses in the base of the induction, it is not possible to shatter $2\alpha_{i+1}+\beta_{i+1}$ points using just $\alpha_{i+1}$ (resp., $\beta_{i+1}$) closed (resp., open) intervals on the $(i+1)$-th dimension. 

Hence it is not possible to shatter $X$ and we have a contradiction.
\end{proof}

\section{Bounding the VC dimension of a given query class}\label{app:reduction}
Algorithm~\ref{alg:reduction} operates ``\emph{consolidating}'' redundant clauses in and equivalent minimal set. 
Such consolidation is achieved by first (lines 2-10) determining the minimal open intervals for each feature, and then by merging overlapping closed intervals whenever possible. Note that given a specific choice of a query class, Algorithm~\ref{alg:reduction}, needs to be run just once to bound its VC dimension.
\begin{algorithm}
\caption{Reduction to non-redundant intervals}\label{alg:reduction}
\begin{algorithmic}[1]
\Procedure{ReduceIntervals}{}
\Statex \textbf{Input:} A conjunction of k clauses expressed as open or closed intervals
\Statex \textbf{Output:} An equivalent non-redundant conjunction of clauses
\State $C\leftarrow False$ \Comment{Initialization non redundant intervals}
\If{ Any constraints of the kind ``$c\leq a_i$'' given as input}
\State $l_{\leq}\rightarrow \max a_i$ such that $c \leq a_i$ is clause;
\Else
\EndIf
\State $r_{\geq}\rightarrow \min a_i$ such that $c \geq a_i$ is clause;
\State $l_{<}\rightarrow \max a_i$ such that $c \neq a_i$ is clause;
\State $r_{>}\rightarrow \min a_i$ such that $c \neq a_i$ is clause;
\State $I\leftarrow$ list of closed intervals $(a_i\leq c \leq b_i)$ sorted according to the values of the $a_i$ increasingly.
\State $op_l = \leq$
\If{$\exists (a_i,b_i)\in I|a_i\leq \max\{l_{\leq}, l_<\}$}
\State $l'\leftarrow \max\{b_i|(a_i,b_i)\in I\text{ and } a_i\leq \max\{l_{\leq}, l_<\}\}$;
\Else
\State $l'\leftarrow \max\{l_{\leq}, l_<\}$;
\If{$l'\neq l_{\leq}$}
\State $op_l \leftarrow <$
\EndIf{}
\EndIf{}
\State $op_r = \geq$
\If{$\exists (a_i,b_i)\in I| b_i\geq \min\{r_{\geq}, r_>\}$}
\State $r'\leftarrow \min\{a_i|(a_i,b_i)\in I\text{ and } b_i\geq \min\{r_{\geq}, r_>\}\}$;
\Else
\State $r'\leftarrow \min\{r_{\geq}, r_>\}$;
\If{$r'\neq r_{\geq}$}
\State $op_r \leftarrow >$
\EndIf{}
\EndIf
\If{$l'>r'$}
\State \textbf{return} $C\leftarrow \texttt{True}$
\ElsIf{$l'=r'\text{ and }\left((op_l = \leq )\text{ or }(op_r=\geq)\right)$}
\State \textbf{return} $C\leftarrow \texttt{True}$
\Else
\State $C\leftarrow (c\ op_l\ l')\vee (c\ op_r\ r')$
\EndIf
\State Update $I$ by removing all intervals $(a_i,b_i)\text{ such that }b_i\leq l'\text{ or }a_i\geq r'$;
\State {$I'\leftarrow \emptyset$}
\While{$I \neq \emptyset$} \Comment{Remove intervals contained in larger intervals}
\State $a \leftarrow \min \{a_i|(a_i,b_i)\in I\}$
\State $b\leftarrow \max\{b_i|(a,b_i)\in I\}$

\State $I'\leftarrow I' \cup \{(a,b)\}$
\State $I\leftarrow I\setminus \{(a_i,b_i)| a_i\geq a\text{ and }b_i\leq b\}$
\EndWhile
\While{$I' \neq \emptyset$} \Comment{Merge intervals with superpositions}
\State $a\leftarrow \min \{a_i|(a_i,b_i)\in I'\}$
\State $b\leftarrow b | (a,b) \in I')$
\If{$\exists b_i | (a_i,b_i)\in I'\text{ and }a_i\leq b$ }
\State $b'\leftarrow \max\{b_i|(a_i,b_i)\in I'\text{ and }a_i\leq b\}$
\State $I\leftarrow I\setminus \{(a_i,b_i)| a_i\geq a\text{ and }b_i\leq b'\}$
\State $I'\leftarrow I' \cup \{(a,b')\}$
\Else{}
\State $I'\leftarrow I' \cup \{(a,b)\}$
\State $C \leftarrow C \vee (a\leq c\leq b)$
\EndIf{}
\EndWhile
\State \textbf{return} $C$ \Comment{Output non-redundant constraints}
\EndProcedure
\end{algorithmic}
\end{algorithm}

\section{Modified \chisqtest}
\label{sec:modifiedchisq}
As described in \cite{chi2book} by modifying the \chisqtest to 
test a more complex null hypothesis
\begin{equation}
H_0: \sum_{k=1}^K \frac{\left( p_k - \hat{p}_k \right)^2}{p_k} < \epsilon_{\chi^2}
\end{equation}
together with the corresponding test statistic 
\[
T = n\gamma \sum_{k=1}^K \frac{\left( p_k - \hat{p}_k \right)^2}{p_k}
\]
following a non-central $\chi^2$ distribution with $K-1$ degrees of freedom and non-locality parameter 
\[
\lambda_{\mathrm{nc}} = n \gamma \epsilon_{\chi^2} = m \epsilon_{\chi^2}.
\]
The distance guaranteed then is the $\chi^2$-distance $d_{\chi^2}$. In fact, this testing procedure can be further generalized to use other distances that belong to the family of f-divergences like the total variation distance or Kullback-Leibler divergence via an suitable transformation or by using non-parametric models to estimate the distance measures before comparing them via statistical testing as in described in more detail in \cite{fdivergence}.

\subsection{VC bounds and Chernoff-Hoeffding bounds}\label{app:vcchernoff}
When only a single visualization needs to be recommended, using Chernoff-Hoeffding bounds is both easier and yields a better guarantee than a VC-based bound. However, when testing multiple visualizations on the same data, the VC approach dominates Chernoff-Hoeffding bounds.
\begin{figure}[ht!]
\centering
     \includegraphics[width=.7\columnwidth]{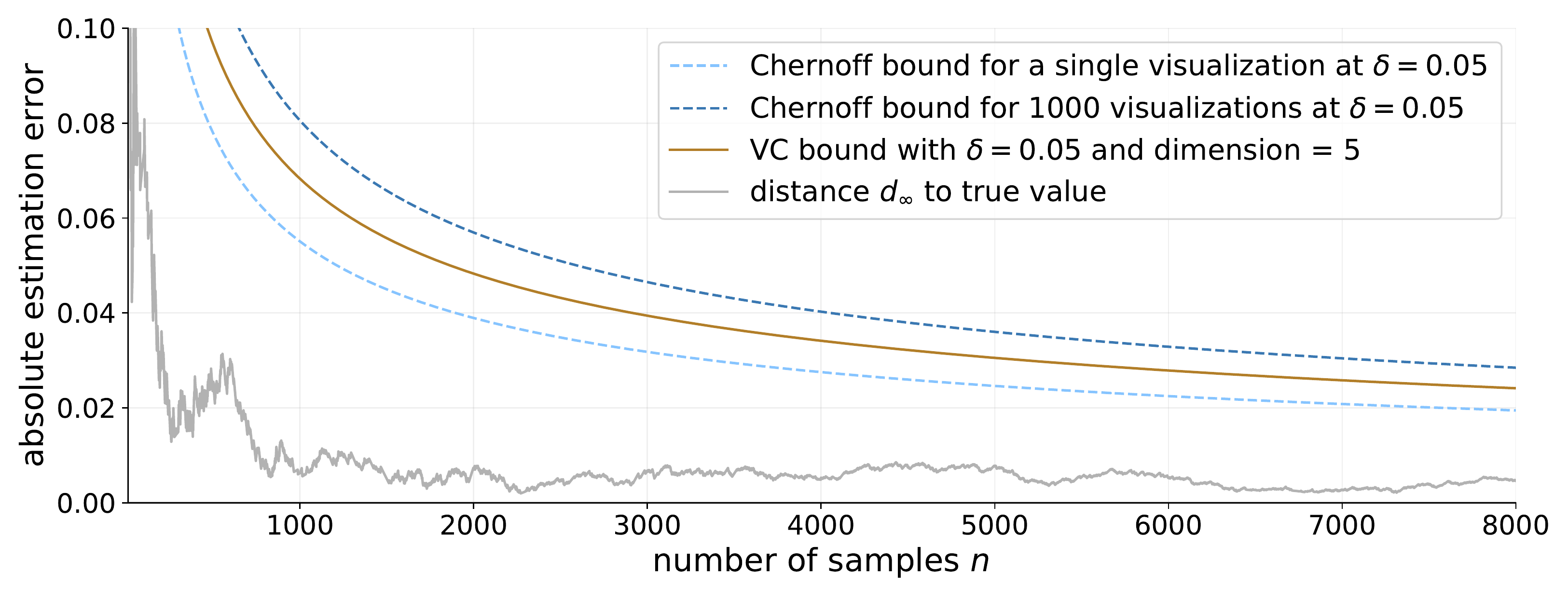}
\caption{Observed estimation error for a single random experiment and bounds obtainable at a significance level of $\delta = 0.05$. For a VC dimension of $d=1$, the bound is only slightly worse than the Chernoff bound of a single visualization. However, with increasing number of visualizations Chernoff bounds are more conservative than bounds obtained via VC.}
\label{fig:exp4_chernoffvsvc}
\end{figure}

To demonstrate this, random data was generated to estimate the probability mass function of a biased Binomial distributed according to $\sim \mathrm{Bin}(10, 0.3)$. In \autoref{fig:exp4_chernoffvsvc} a sample path is shown together with the theoretical bounds for a setup with a VC dimension of $5$. Comparing to a Chernoff-Hoeffding bound of a single visualization, the VC approach outperforms Chernoff-Hoeffding bounds when learning multiple visualizations at once.

\subsection{Restricting the search space}
In this experiment we show that without restriction of the search space, it is likely that only few visualizations get recommended. For this, we took again the survey dataset and removed constants and identifier alike columns. To make it visually more distinct, we also added some artificial columns like a running identifier. 
\begin{figure}[H]
    \centering
    \includegraphics[width=.6\columnwidth]{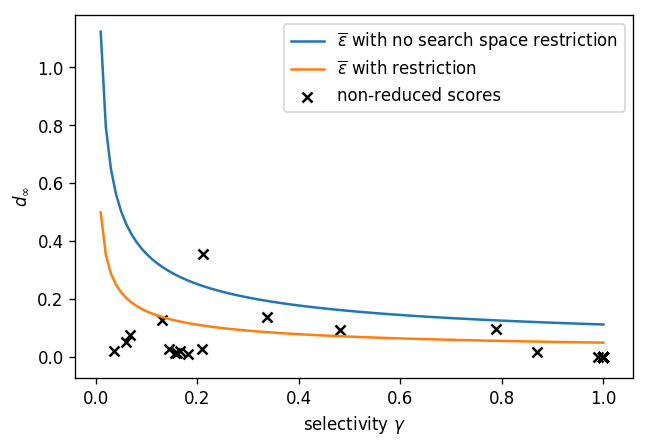}
    \caption{$\overline{\epsilon}$-curves before and after restricting the search space. The distance $d_\infty$ needs to be higher than the uncertainty quantified via $\overline{\epsilon}$.}
    \label{fig:searchspace_restriction}
\end{figure}
As shown in \autoref{fig:searchspace_restriction}, restricting the search space leads to more discoveries. I.e. the goal is to not be over-conservative and apply meaningful preprocessing or feature selection by the user first.


\newpage